\documentclass[journal]{IEEEtran}
\usepackage{graphicx}     
\usepackage{epstopdf}
\usepackage{amsmath}
\usepackage{amsthm}
\usepackage{algcompatible}
\usepackage{amssymb}
\usepackage[usenames, dvipsnames]{color}
\usepackage{algorithm}
\usepackage{algpseudocode}
\usepackage{mathtools}
\usepackage{breqn}
\usepackage{bbm}
 \usepackage{cite}
\graphicspath{{./figures/}}
\usepackage{pdfpages}
\newtheorem{definition}{Definition}
\newtheorem{assumption}{Assumption}
\newtheorem{remark}{Remark}
\newtheorem{theorem}{Theorem}
\newtheorem{lemma}{Lemma}
\newtheorem{proposition}{Proposition}
\newtheorem{corollary}{Corollary}
\newtheorem*{problem}{Problem}
\newtheorem{example}{Example}

\newcommand{\norm}[1]{\ensuremath{\left\| #1\right\|}}
\newcommand{\paren}[1]{\ensuremath{\left( #1\right)}}
\newcommand{\clint}[1]{\ensuremath{\left[ #1\right]}}

\newcommand{\set}[1]{\ensuremath{\left\{ #1\right\}}}
\newcommand{\vct}[1]{\ensuremath{\boldsymbol{#1}}}  

\newcommand{\matr}[1]{\ensuremath{\clint{\begin{array} #1 \end{array}}}}
\newcommand{\expe}[1]{\ensuremath{\mathbb{E}\set{#1}}}
\newcommand{\cov}[1]{\ensuremath{\mathrm{Cov}\set{#1}}}
\newcommand{\var}[1]{\ensuremath{\mathrm{Cov}\set{#1}}}
\newcommand{\infos}[2]{\ensuremath{\mathcal{I}^{#1}_{#2}}}
\newcommand{\info}[1]{\ensuremath{\mathcal{I}_{#1}}}
\newcommand{\x}[1]{\ensuremath{\vct{x}_{0:#1}}}
\newcommand{\ifo}[1]{\ensuremath{\mathcal{J}_k\paren{#1}}}
\newcommand{\EV}{\ensuremath{\mathcal{C}}}
\newcommand{\event}{\ensuremath{\mathcal{B}}}

\newcommand{\y}[1]{\ensuremath{z_{#1}}}
\newcommand{\g}[1]{\ensuremath{\vct{g}_{0:#1}}}
\renewcommand{\gg}[1]{\ensuremath{\tilde{\vct{g}}_{0:#1}}}

\newcommand{\1}{\ensuremath{\mathbbm{1}}}

\newcommand{\est}{\ensuremath{\eta}}
\newcommand{\estm}{\ensuremath{H}}
\newcommand{\REAL}{\ensuremath{\mathbb{R}}}

\DeclareMathOperator*{\argmin}{arg\,min}

\DeclareMathOperator{\Tr}{\mathbf{tr}}

\usepackage{array}
\newcolumntype{"}{@{\hskip\tabcolsep\vrule width 2pt\hskip\tabcolsep}}

\begin{document}
\title{State-Secrecy Codes for \\ Networked Linear Systems}

\author{Anastasios~Tsiamis,~\IEEEmembership{Student Member,~IEEE,}
        Konstantinos~Gatsis,~\IEEEmembership{Member,~IEEE,}
        and~George~J.~Pappas,~\IEEEmembership{Fellow,~IEEE}
\thanks{The authors  are   with   the   Department   of   Electrical   and   Systems  Engineering,  University  of  Pennsylvania,  Philadelphia,  PA  19104.
		 Emails: \{atsiamis,kgatsis,pappasg\}@seas.upenn.edu
}
}

\maketitle
\begin{abstract}
In this paper, we study the problem of remote state estimation, in the presence of a passive eavesdropper. An authorized user estimates the state of an unstable linear plant, based on the packets received from a sensor, while the packets may also be intercepted by the eavesdropper. Our goal is to design a coding scheme at the sensor, which encodes the state information, in order to impair the eavesdropper's estimation performance, while enabling the user to successfully decode the sent messages. 
We introduce a novel class of codes, termed State-Secrecy Codes, which use acknowledgment signals from the user and apply linear time-varying transformations to the current and previously received states. By exploiting the properties of the system's process noise, the channel physical model and the dynamics, these codes manage to be fast, efficient and, thus, suitable for real-time dynamical systems.  
We prove that under minimal conditions, State-Secrecy Codes achieve perfect secrecy, namely the eavesdropper's estimation error grows unbounded almost surely, while the user's estimation performance is optimal. 
These conditions only require that at least once, the user receives the corresponding packet while the eavesdropper fails to intercept it. 
Even one occurrence of this event renders the eavesdropper's error unbounded with  asymptotically optimal rate of increase.
State-Secrecy Codes are provided and studied for two cases, i) when direct state measurements are available, and ii) when we only have output measurements.
The theoretical results are illustrated in simulations.
\end{abstract}


\IEEEpeerreviewmaketitle

\section{Introduction}
The recent emergence of the Internet of Things (IoT) as a collection of interconnected sensors and actuators has created a new attack surface for adversarial attacks \cite{cardenas2008research}, \cite{sandberg2015cyberphysical}. Research efforts in the context of control systems have targeted denial-of-service attacks \cite{gupta2010optimal} and data integrity of compromised sensors \cite{pajic2014robustness}, \cite{mo2014resilient}, \cite{fawzi2014secure}, \cite{pasqualetti2015control}.  
However, another fundamental vulnerability of such interconnected systems is eavesdropping attacks, especially when the underlying medium of communication is of a broadcast nature, i.e. as in wireless systems \cite{Survey_Wireless_Security}. Eavesdroppers are usually passive adversaries, which intercept unauthorized data. This data leakage, not only compromises confidentiality, but could be also used to perform more complex attacks \cite{teixeira2015secure_resource}. 

In this paper, 
we study eavesdropping attacks in the context of real-time dynamical systems,  where the source of information is a dynamical system. In many IoT applications, sensors collect state information about the dynamical system and send it to an authorized user, i.e. a controller, a cloud server, etc. through a (wireless) channel, see Figure \ref{Figure_ProblemAbstract}. Our goal is to design codes such that the user receives the confidential state information, while any eavesdroppers are confused about the true state.  One of the main challenges in designing such codes for real-time dynamical systems is time itself; in general, it is desirable to avoid elaborate codes which might introduce severe delays to the data processing of the user. 
\begin{figure}[t] \centering{
		\includegraphics[scale=0.61]{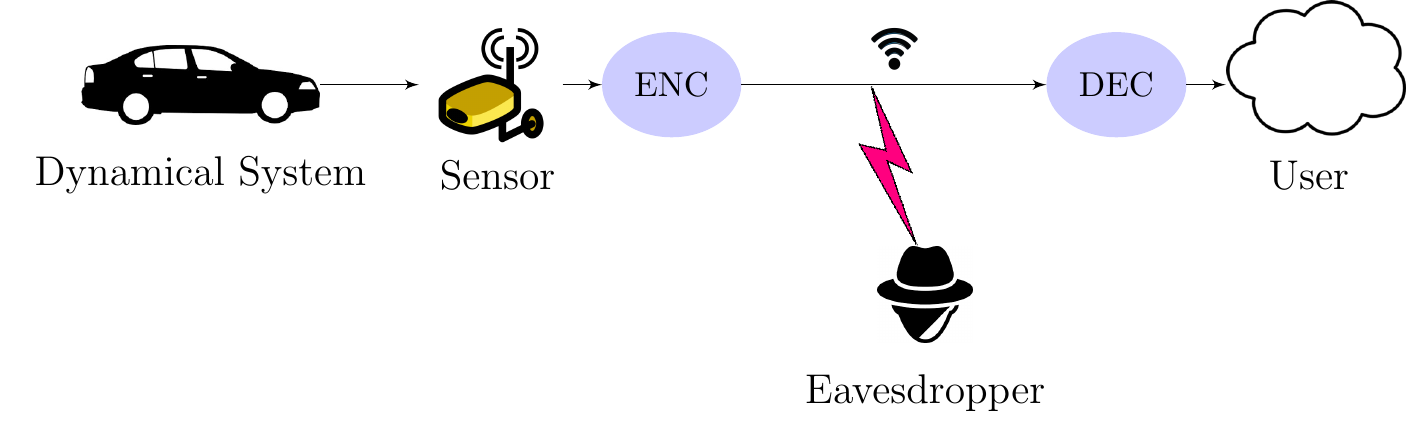}  
		\caption {This figure presents an abstract description of our problem. A sensor sends confidential information about the state of a dynamical system to an authorized user (i.e. a controller, the cloud) over a wireless channel. Meanwhile, an eavesdropper might intercept the sent messages. Our goal is to design encoder-decoder pairs such that confidentiality is protected.}
		\label{Figure_ProblemAbstract} }
\end{figure}

Most of the current defense mechanisms involve cryptography-based tools \cite{Cryptography} against computationally limited adversaries. However, these are generic tools that do not exploit the properties of the underlying physical process, i.e. the communication channel or the system's dynamics. Moreover, encryption may also introduce computation and communication overheads, which are significant for real-time systems \cite{lee2010price}. This motivates the development of other approaches for reinforcing confidentiality.

One such approach includes developing defense mechanisms in the physical layer of wireless communications \cite{Proc_IEEE_Special_issue, bloch2008wireless}. 
These methods exploit the characteristics of the underlying channel used for communication, e.g., the wireless medium, to offer provable guarantees about confidentiality against eavesdroppers. Information-theoretic approaches \cite{Wyner_wiretap, LiangPoor2008Secure, Secrecy_MIMO} define a notion of secrecy capacity of channels and give conditions about the existence of codes such that the eavesdropper receives no information. However, finding such codes is challenging in practice and requires knowledge of the eavesdropper's channel, which may not be available. Nonetheless, in the case of packet erasure channels, those erasures can be exploited to randomly create secrets, thus leading to more practical defense mechanisms~\cite{safaka2016creating}. 
Still, the aforementioned approaches typically involve static sources which is not the case in dynamical systems.
Recently, in \cite{li2011communication}, \cite{Wiretap_estimation}, \cite{wiese2016secure}, information-theoretic tools were applied in the case of dynamical systems, to the problem of remote estimation in the presence of eavesdroppers. 
As with the static case, 
those approaches also face the challenge of finding codes and requiring knowledge about the eavesdropper's channel. 

A control-theoretic approach was recently employed in \cite{tsiamis2016state}, \cite{leong2017remote}, 
where the performance metric of both the user and the eavesdropper is minimum mean square error (mmse). This framework does not use any encoding. Instead, a secrecy mechanism is employed, which withholds information, either randomly \cite{tsiamis2016state} or deterministically \cite{leong2017remote}. In the case of unstable systems, under certain conditions, the eavesdropper's expected error can grow to infinity  while the user's expected error remains bounded. However, the main disadvantage of those works is that the guarantees about the eavesdropper are in expectation, not almost surely, while the user's performance is degraded as a side effect.

In this paper, we develop a novel class of codes, suitable for real-time dynamical systems, which we call State-Secrecy Codes. When direct state measurements are available,  the system's state is encoded by subtracting a weighted version of the user's most recently received state from the current state. This only requires acknowledgment signals from the user back to the sensor.  By exploiting the inherent process noise of the physical system, the channel's randomness as well as the dynamics, State-Secrecy Codes achieve very strong guarantees with very low computational cost. 
In particular, confidentiality is guaranteed if at some time the user receives the encoded state, while the eavesdropper fails to intercept it. Due to our code, a single occurrence of this event, which we call critical event, makes the eavesdropper lose important information about the system state. Then, as time passes, the dynamics amplify the uncertainty of the eavesdropper created by this information loss, regardless of the eavesdropper's computational capabilities.

 In Section \ref{Section_Formulation} we present the problem formulation. The dynamical system is modeled as linear, while the channel, is modeled as a packet drop channel. Similar to \cite{Wiretap_estimation, wiese2016secure ,tsiamis2016state}, we assume that the system is unstable (see Section \ref{Section_Formulation} for discussion). Both the user and the eavesdropper know the coding scheme and use minimum mean square error estimators as decoders. We also introduce a novel control-theoretic notion of perfect secrecy, requiring that the eavesdropper's error grows unbounded almost surely, while the user's performance is optimal. 
In Section \ref{Section_Perfect}, where direct state measurements are available, we show that by employing State-Secrecy Codes, perfect secrecy can be guaranteed under remarkably mild conditions (Theorem \ref{THM_perfect_secrecy}). It is sufficient for the critical event, where the user receives the corresponding packet while the eavesdropper misses it, to occur at least once. Even one occurrence renders the eavesdropper's error unbounded with asymptotically optimal rate of increase (Corollary \ref{THM_rate_of_increase}).
In Section \ref{Section_Output}, we extend the results to the case of output measurements (Theorem \ref{THM_output}). The sensor performs local Kalman filtering before applying a State-Secrecy Code to the local estimates. 

In summary,
our main contributions, are the following: 
\begin{itemize}
	\item We introduce State-Secrecy Codes, which are fast and efficient, thus, suitable for real-time dynamical systems. Their efficiency does not depend on the computational capabilities of the eavesdropper, contrary to encryption techniques \cite{Cryptography}.
	\item The codes  achieve unbounded eavesdropper's error almost surely, while the user's estimation performance is optimal. This supersedes the results in \cite{tsiamis2016state, leong2017remote} where unbounded eavesdropper's error is achieved only in expectation and the user's performance is degraded.
	\item The condition for perfect secrecy is remarkably minimal, requiring just a single occurrence of the critical event. 
	\item The assumptions about the channel are minimal. Thus, our results are distribution free.
\end{itemize}

We conclude this paper by illustrating the performance of the State-Secrecy Codes in simulations in Section \ref{Section_Simulations}, and with remarks in Section \ref{Section_Conclusion}. All proofs are included in the Appendices.

\section{Problem formulation}\label{Section_Formulation}
\subsection{Dynamical system model}
The considered remote estimation architecture is shown in Figure \ref{Figure_ProblemSetup}  and consists of a sensor observing a dynamical system, a legitimate user, and an eavesdropper. The dynamical system is linear and has the following form:
\begin{align} \label{EQN_system}
x_{k+1}&=Ax_{k}+w_{k+1}\\
y_{k}&=Cx_{k}+v_{k}\label{EQN_output}
\end{align}
where $x_{k}\in \REAL^{n}$ is the state, $y_k\in \REAL^m$ is the output and $k\in \mathbb{N}$ is the (discrete) time. Matrices $A\in\REAL^{n\times n}$ and $C\in \REAL^{m\times n}$ are the system and output matrices respectively. Signals $w_{k}\in \REAL^{n}$ and $v_k\in \REAL^{m}$ are the process and measurement noises respectively and are modeled as independent Gaussian random variables with zero mean and covariance matrices $Q$ and $R$ respectively. The initial state $x_{0}$ is also a Gaussian random variable with zero mean and covariance $\Sigma_{0}$ and is independent of the noise signals. 
Matrices $Q,\,R,\,\Sigma_{0}$ are assumed to be 
positive definite (unless otherwise stated). In more compact notation $Q,\,R,\,\Sigma_{0}\succ 0$, where $\succ$ ($\succeq$) denotes comparison in the positive definite (semidefinite) cone.   We also assume that the pair $\paren{A,C}$ is detectable. 
All system and noise parameters $A, C, Q, R, \Sigma_0$  are assumed to be public knowledge, available to all involved entities, i.e., the sensor, the user, and the eavesdropper. 
We also note that we consider a common probability space $\Omega$ for all random quantities (noises, initial condition and channel outcomes). 
Throughout this paper, we assume the system is unstable. 
 \begin{assumption}\label{ASSUM_unstable}
 	We assume that the dynamical system \eqref{EQN_system} is unstable  i.e., its spectral radius is $\rho(A) = \max_{i} |\lambda_i(A)|>1$. 
\end{assumption}
From a security point of view, we can achieve much better confidentiality when the system is unstable; without any measurements, the unstable dynamics amplify the uncertainty caused by the process noise. On the other hand, when the system is stable, the problem is more challenging. The eavesdropper can always predict that a stable system is close to equilibrium without even eavesdropping. We do not deal with stable systems in this paper, but it is subject of ongoing work.

\begin{figure}[t] \centering{
		\includegraphics[scale=0.75]{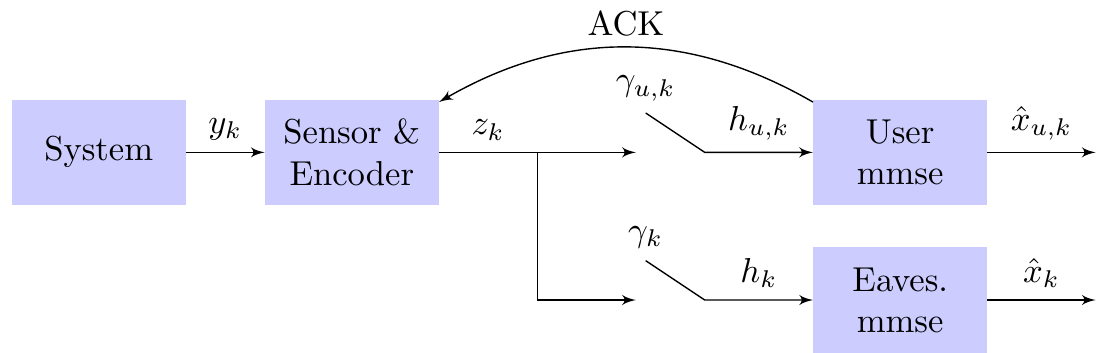}  
		\caption {A sensor collects the output $y_k$ of the dynamical system \eqref{EQN_system}. Then it transmits an encoded version  $\y{k}$ of the output to the channel, which is neither reliably nor securely received by the user. The packets might be dropped, as captured by $\gamma_{u,k}$, and might be intercepted by the eavesdropper, as captured by $\gamma_{k}$. To decode the messages, the user and the eavesdropper compute the minimum mean square error (mmse) estimates $\hat{x}_{u,k}$ and $\hat{x}_{k}$ respectively.}
		\label{Figure_ProblemSetup} }
\end{figure}

\subsection{Channel model}
The sensor communicates over a channel with two outputs/receivers as shown in Figure \ref{Figure_ProblemSetup}.  The input to the channel is denoted by $\y{k}\in\REAL^{n}$.  
The first output, denoted by $h_{u,k}$, is the authorized one to the user, while the second, denoted by $h_{k}$, is the unauthorized one to the eavesdropper.
Communication follows the packet-based paradigm commonly used in networked control systems \cite{Sinopoli2004Kalman}, \cite{hespanha2007survey}. We assume that the packets consist of sufficiently large number of bits to neglect quantization effects~\cite{Sinopoli2004Kalman, hespanha2007survey, GatsisEtal14}.

Communication with the user is unreliable, i.e., may undergo packet drops. Additionally, communication is not secure against the eavesdropper, i.e., the latter may intercept transmitted packets. 
In particular, we denote by $\gamma_{u,k}\in\set{0,1}$ the outcome of the user packet reception at time $k$, and by $\gamma_{k}\in\set{0,1}$ the outcome of the eavesdropper's packet interception. When $\gamma_{u,k}=1$ (or $\gamma_{k}=1$), then the reception (interception) is successful. Otherwise the reception (interception) is not successful and the respective packet is dropped.  Thus, the outputs of the channel are modeled as:
\begin{equation}\label{EQN_channel_model} 
\begin{aligned}
h_{u,k}&=\left\{
\begin{aligned}\y{k}, \quad &\text{ if } \gamma_{u,k}=1\\
\varepsilon,  \quad &\text{ if } \gamma_{u,k}=0
\end{aligned}\right.\\
h_{k}&=\left\{
\begin{aligned}\y{k}, \quad &\text{ if } \gamma_{k}=1\\
\varepsilon,  \quad &\text{ if } \gamma_{k}=0
\end{aligned}\right.
\end{aligned}
\end{equation}
for $i=1$, $2$, where
symbol $\varepsilon$, is used to represent the ``no information" outcome. The variables $\{\gamma_{u,k},\gamma_{k},\, k=0,1,\dots\}$ are modeled as random and assumed to be independent of the initial state $x_0$, the process noise $w_k$, and the measurement noise $v_k$, for $k=0,1,\dots$. We do not assume any specific joint distribution of the channel outcomes; as explained in Section \ref{Section_Perfect} the results of this paper (Theorems~\ref{THM_perfect_secrecy}, \ref{THM_output}) are distribution-free.

In addition to the main channel, the user can reliably send  acknowledgment signals back to the sensor via the reverse channel. Thus, at any time step the sensor knows what is the latest received message $\y{k}$ at the user. 
Meanwhile, we assume that the eavesdropper is able to intercept all acknowledgment signals, and thus, knows the history of user's packet successes. In that respect, we model a powerful eavesdropper. 
On the other hand, neither the sensor nor the user have any knowledge about the eavesdropper's intercept successes $\gamma_{k}$.
\subsection{Encoder definition}
The sensor collects output measurements $y_k$ and encodes them by sending $\y{k} \in \mathbb{R}^p$ over the channel at each time step $k$, where $p$ is an integer to be designed. Hence, under this definition and by \eqref{EQN_channel_model}, the channel outputs $h_{u,k}$, $h_{k}$ take values in $\REAL^{p}\cup\set{\varepsilon}$.
In general, the encoder may produce $\y{k}$, given all the information at the sensor at time $k$, i.e. current and past outputs  $y_{t}$ for $t\leq k$, past sent messages $\y{t}$ for $t< k$, as well as past user's channel outcomes $\gamma_{u,t}$ for $t < k$:
\begin{equation}\label{EQN_DEF_coding}
z_k=f_k\paren{y_k,y_t,z_t,\gamma_{u,t},t<k},
\end{equation}
where $f_k$ is a function from $\REAL^{m(k+1)+pk}\times \set{0,1}^{k}$ to $\REAL^{p}$. 
Although this allows encoders with infinite memory, our proposed one (see Sections \ref{Section_Perfect}, \ref{Section_Output}) does not need the whole history and, thus, only uses finite memory. 
\subsection{MMSE Estimation}
Both the user and the eavesdropper know the encoding scheme
and use the minimum mean square error (mmse) estimate to decode the received/intercepted messages. 
This estimate depends on their information up to time $k$. So, we define the batch vector of received channel outputs $\vct{h}_{u,0:k}=\paren{h_{u,0}, \ldots, h_{u,k}}$ and channel outcomes $\vct{\gamma}_{u,0:k}=\paren{\gamma_{u,0},\dots,\gamma_{u,k}}$ for the user. Similarly, for the eavesdropper, we define $\vct{h}_{0:k}=\paren{h_{0}, \ldots, h_{k}}$ and  $\vct{\gamma}_{0:k}=\paren{\gamma_{0},\dots,\gamma_{k}}$.
Then, the user's information at time $k$ is denoted by $\mathcal{I}^u_k=\{\vct{h}_{u,0:k} \}$, with $\mathcal{I}^u_{-1}=\emptyset$. Respectively, we denote the eavesdropper's information at time $k$ by
\begin{equation}\label{EQN_eavesdropper_information}
\info k=\set{\vct{h}_{0:k},\vct{\gamma}_{u,0:k}},\,\info{-1}=\emptyset
\end{equation} 
  The two information sets are not symmetric, i.e. the eavesdropper has the additional information of the user's reception success history.
With those definitions,  the mean square error estimate, $\hat{x}_{u,k}$, at the user and the respective estimation error covariance matrix $P_{u,k}$ are given by:
\begin{align}\label{EQN_estimate_definition_user}
\hat{x}_{u,k}&=\expe{ x_{k}\vert\infos u k},\quad
P_{u,k}=\var{x_{k}\vert \infos u k}
\end{align}
where the conditional covariance of any random vector $Z$ with respect to some other random vector $\mathcal{I}$ is defined as 
$$\var{Z|\mathcal{I}}=\expe{\paren{Z-\expe{Z|\mathcal{I}}}\paren{Z-\expe{Z|\mathcal{I}}}'|\mathcal{I}}.$$
Similarly, the eavesdropper's mean square error estimate, $\hat{x}_k$ and the respective estimation covariance matrix $P_k$ are:
\begin{align}\label{EQN_estimate_variance_eavesdropper}
\hat{x}_{k}=\expe{ x_{k}\vert \info k},\quad P_{k}=\var{x_{k}\vert\info k}.
\end{align}

\subsection{Problem}
The goal of this work is to design the encoding scheme at the sensor, so that \emph{perfect secrecy} is achieved, introduced in the following definition. We require the eavesdropper's estimation error to grow unbounded, while the user successfully decodes the information and has optimal estimation performance. 
The user's estimation scheme is optimal when at the successful reception times, the estimate and estimation error covariance are the same as if no packet had been dropped in the past (see also \cite{gupta2007optimal}). 
\begin{definition}[Perfect Secrecy]\label{DEF_perfect_secrecy}
	Given system \eqref{EQN_system}, \eqref{EQN_output} and channel model \eqref{EQN_channel_model}, we say that a coding scheme \eqref{EQN_DEF_coding} achieves perfect secrecy if and only if both of the following conditions hold:
	\renewcommand{\labelenumi}{(\roman{enumi})}
	\begin{enumerate}
		\item the user's performance is optimal: \begin{equation}\label{EQN_optimal_estimation_user}
		\left. \begin{aligned}
		\hat{x}_{u,k}&=\expe{x_k|\vct{y}_{0:k}}\\
		P_{u,k}&=\var{x_k|\vct{y}_{0:k}}
		\end{aligned}\right\},\,\text{when }\gamma_{u,k}=1,
		\end{equation}
		where $\vct{y}_{0:k}=\paren{y_0,\dots,y_k}$. 
		\item the eavesdropper's error diverges to infinity with probability one:	\begin{equation}\label{EQN_Definition_Perfect_Secrecy_1}
		\Tr{P_{k}}\stackrel{a.s.}{\rightarrow}\infty\footnote{The limit $\stackrel{a.s.}{\rightarrow}$ denotes almost sure convergence \cite{durrett2010probability} with respect to the common probability space $\Omega$, as $k\rightarrow \infty$.},
		\end{equation}
		where  $\Tr$ is the trace operator. 
	\end{enumerate}
\end{definition}
This notion of secrecy is asymptotic, which is an inherent property of the problem. Even without any interceptions, the eavesdropper can maintain the trivial open-loop prediction estimate, i.e. $\hat{x}_k=0$, that has unbounded but finite estimation error at any time $k$.
 Moreover, we remark that \eqref{EQN_Definition_Perfect_Secrecy_1} guarantees aggregate state secrecy in that, at least one but not necessarily all eigenvalues of the eavesdropper's error covariance grow unbounded, e.g., she might still be able to estimate a stable part of the state with bounded error (see also Section \ref{Section_Simulations}).  With this definition, we formally present the problem that we solve in this paper.
\begin{problem}
 	Given system \eqref{EQN_system}, \eqref{EQN_output} and channel model \eqref{EQN_channel_model}, design a coding scheme \eqref{EQN_DEF_coding} such that perfect secrecy is achieved, as described in Definition \ref{DEF_perfect_secrecy}.
\end{problem}
In the following sections, we consider present and analyze State-Secrecy Codes for two cases. In Section \ref{Section_Perfect}, we ignore the output model ($C=I$, $R=0$) and assume that the sensor measures the state perfectly. In Section \ref{Section_Output}, we include the output model (general $C$, $R\succ 0$). In both cases, perfect secrecy is achieved by exploiting the acknowledgment signals, the unstable system dynamics, the process noise, as well as the randomness of the channel. 
\section{Perfect Secrecy with State Measurements}\label{Section_Perfect}
In this section, we introduce State-Secrecy Codes for the case of direct state measurements ($C=I$ and $R=0$). 
Informally, the sensor encodes and transmits the current state measurement $x_k$ as a weighted state difference of the form $x_k-A^{k-t_k}x_{t_k}$, where $x_{t_k}$ is a previous state called the \emph{reference state} of the encoded message, for some $t_k<k$ depending on $k$. The sensor and the user can agree on this reference state via the acknowledgment signals, e.g., it can be the most recent state received at the user's end. At the user's side, no information is lost with this encoding; upon receiving a new message $x_k-A^{k-t_k}x_{t_k}$, she can first recover $x_{k}$ by adding $A^{k-t_k}x_{t_k}$ and then notify the sensor to use $x_k$ as the reference state for the next transmission.
 
On the other hand, on the event that the eavesdropper fails to intercept that reference packet $x_{t_k}$ at time $t_k$, her error starts increasing. That is because the eavesdropper misses the reference state $x_{t_k}$ and, thus, cannot decode a following packet of the form $x_k-A^{k-t_k}x_{t_k}$ to obtain $x_k$. But this also obstructs the eavesdropper from decoding future packets, as any following reference state $x_k$ for some $k>t_k$, depends on the current reference state $x_{t_k}$ and so on. This triggers an irreversible chain reaction effect,  which combined with the unstable system dynamics, leads to an exponentially growing eavesdropper's estimation error. For this reason, we call this event, where the user receives a packet at time $t_k$ while the eavesdropper misses it, the \emph{critical event}.

The following definitions formally describe our coding scheme.
	We define the \emph{reference time}  $t_{k}$ to be the time of the most recent successful reception at the user before $k$:
	\begin{equation}\label{EQN_reference_time}
	t_{k}=t_{k}(\vct{\gamma}_{u,0:{k-1}})=\max\set{0\le t<k:\gamma_{u,k}=1}.
	\end{equation}
	Until the first successful transmission, i.e. when the set $\max \set{0\le t< k:\: \gamma_{u,t}=1}$ is empty, we use the convention $t_{k}=-1$.
	Respectively, $x_{t_{k}}$ is called the \emph{reference state}, with $x_{-1}=0$.
\begin{definition}[State-Secrecy Code]\label{DEF_name_codes}
	Given the unstable system matrix $A$ in \eqref{EQN_system}, a State-Secrecy Code applies the following time-varying linear operation \begin{equation} \label{EQN_Coding}
	\y{k}=x_k-A^{k-t_k}x_{t_k},
	\end{equation}
	where $t_k$ is the reference time defined in \eqref{EQN_reference_time}.\hfill $\diamond$
\end{definition}
The intuition about selecting the weighting factor $A^{t-t_k}$ can be found in Remark~\ref{REM_intuition}. The implementation of the scheme is described in Algorithm~\ref{ALG_coding}. The sensor keeps in memory the reference time $t_k$ and state $x_{t_k}$, with $t_{0}=-1$, $x_{-1}=0$. At each time $k$,  it transmits $z_k$ as in \eqref{EQN_Coding}.
If the user receives the packet successfully, it sends an acknowledgment signal back to the sensor. In this case, the sensor updates the reference time $t_{k+1}=k$. Otherwise, it keeps $t_{k+1}=t_{k}$. 
The memory required for the encoder is minimal ($\mathcal{O}(n)$) and the only computational burden is a matrix-vector multiplication~($\mathcal{O}(n^2)$).
 \begin{algorithm}[!t]
 	\caption{State-Secrecy Code}
 	\label{ALG_coding}
 	\begin{algorithmic}[1] {}
 		\Require $A$ and $\,x_k$ at each $k\ge 0$
 		\Ensure Encoded signals $\y{k}$, for all $k\ge 0$.
 		\Statex Let $t$ represent the time of user's most recent message.
 		\State Initialize $t=-1$, $x_{-1}=0$
 		\For{$k=0,1,\dots$ } 
 		\State{Transmit $\y{k}=x_{k}-A^{k-t}x_{t}$}
 		\If{Acknowledgment received} $t=k$ \EndIf
 		\EndFor
 	\end{algorithmic}
 \end{algorithm}
The critical event formally defined below, is crucial for reinforcing secrecy with our coding scheme.
\begin{definition}[Critical event]\label{DEF_critical_event}
	A critical event occurs at time $k$ if the user receives the packet, while the eavesdropper fails to intercept it:
	\begin{equation}\label{EQN_critical_event}
	\gamma_{u,k}=1,\,\gamma_{k}=0
	\end{equation}
\end{definition}

An example to clarify the coding scheme and the critical event is presented next.

\begin{example} 
	Suppose that for $k=0,1,2,3$ we have the channel outcomes as shown in the first two rows of the following table:
	$$\begin{array} {|l"r|r|c|c|}\hline k& 0&1&2&3\\\hline 
	\text{user }\gamma_{u,k}& 0 &1&1&1\\ \hline \text{eavs. }\gamma_{k} & 1 &0&1&1\\ \hline 
	t_k & -1& -1&1&2\\ \hline 
	z_k & x_0& x_1 &x_2-Ax_1&x_3-Ax_2\\ \hline
	\text{user }h_{u,k} & \varepsilon& x_1 & x_2-Ax_1&x_3-Ax_2\\ \hline
	\text{eavs. }h_{k} & x_0& \varepsilon &x_2-Ax_1&x_3-Ax_2\\ \hline
	\end{array}$$
	Then, the last four  rows of the table are constructed using the definitions of the reference times \eqref{EQN_reference_time}, of the coding scheme \eqref{EQN_Coding}, and the channel outcomes \eqref{EQN_channel_model}.
	Notice that the critical event occurs at time $k=1$, where the user receives $x_1$, while the eavesdropper misses it. Then, the user can  recover $x_2$ at time $k=2$, adding $Ax_1$ to $h_{u,2}$. However, since the eavesdropper does not know $x_1$, she cannot precisely recover $x_2$. Since $\gamma_{u,2}=1$, $x_2$ is the next reference state at time $k=3$. Thus, the eavesdropper will also not be able to recover $x_3$, from $h_{3}=x_3-Ax_2$. Hence, a single occurrence of the critical event impairs future estimation at the eavesdropper.\hfill $\diamond$
\end{example}
Our first result, for the case of state measurements, formally proves the previous observations and is presented in the following theorem. If the critical event $\set{\gamma_{u,k_0}=1,\,\gamma_{k_0}=0}$ occurs at some time $k_0$, then the 
eavesdropper's error starts to grow unbounded exponentially fast. On the other hand, the user's performance is optimal.
\begin{theorem}[Perfect secrecy]\label{THM_perfect_secrecy} 
Consider system \eqref{EQN_system}, with channel model \eqref{EQN_channel_model} and coding scheme \eqref{EQN_Coding}. 
If
\begin{equation}\label{EQN_theorem_condition}
	\mathbb{P}( \gamma_{u,k}=1,\,\gamma_{k}=0,\text{ for some }k\ge 0 )=1, 
\end{equation}
then:
\renewcommand{\labelenumi}{(\roman{enumi})}
\begin{enumerate}
	\item perfect secrecy is achieved according to Definition \ref{DEF_perfect_secrecy}.
	\item conditioned on the event $\set{\gamma_{u,k_0}=1,\,\gamma_{k_0}=0}$ for some $k_0\ge 0$, the eavesdropper's estimation error grows unbounded for $k \geq k_0$ as 
	\begin{equation}\label{EQN_big_o_result}
		\Tr{P_k}\ge c \rho\paren{A}^{2\paren{k-k_0}},
	\end{equation}
\end{enumerate}
where $P_k$ is the error covariance defined in \eqref{EQN_estimate_variance_eavesdropper} and $c>0$ is a constant independent of $k_0$.
\hfill $\diamond$
\end{theorem}  

The above theorem is remarkable as the condition \eqref{EQN_theorem_condition} for perfect secrecy is completely minimal; it only requires the critical event, where the user receives a message without the eavesdropper intercepting it, 
to occur at least once. 
Any joint distribution of packet receptions and interceptions that satisfies this condition is covered, and in this sense the result is channel-free, and holds in cases of practical interest -- see Remark~\ref{REM_channel}.

The proof of Theorem \ref{THM_perfect_secrecy} is included in the Appendix and is a consequence of the following lemma, which can be thought as the worst case, in terms of secrecy, of Theorem \ref{THM_perfect_secrecy}. That is when the critical event $\set{\gamma_{u,k_0}=1,\,\gamma_{k_0}=0}$ occurs at time $k_0$ and the eavesdropper receives all the following packets for $k>k_0$. 
 \begin{lemma}[Worst case analysis]
 \label{THM_worst_case}
Consider system \eqref{EQN_system} with channel model \eqref{EQN_channel_model} and coding scheme \eqref{EQN_Coding}. Suppose that both  events 
\begin{align}\event&=\set{\gamma_{u,k_0}=1,\,\gamma_{k_0}=0}\text{ and}\label{EQN_event_B}\\
\EV&=\set{\gamma_{k}=1,\,\text{for all }k\ge k_0+1}\label{EQN_event_C}
\end{align}
occur for some $k_0\ge 0$.  Then 
	\begin{equation}\label{EQN_Lemma_goal_equation}
		P_{k}=A^{k-k_0}P_{k_{0}}\paren{A'}^{k-k_0}	
		\end{equation}
	for $k\ge k_0$ in $\event\cap \EV$.\hfill $\diamond$
 \end{lemma}
Notice that the equation $P_{k}=A^{k-k_0}P_{k_{0}}\paren{A'}^{k-k_0}$ is unstable with rate $\rho\paren{A}^2$. Hence, even in the most pessimistic scenario for confidentiality, the eavesdropper still has unbounded error. In the general case when the eavesdropper does not intercept all packets after $k_0$, the eavesdropper's error will be even larger (cf. Lemma \ref{THM_comparison_lemma} in Appendix) and verifies the result of Theorem \ref{THM_perfect_secrecy}.

\begin{remark}\label{REM_intuition}
	The choice of $A^{k-t_k}$ is pivotal for achieving secrecy. The difference $x_k-A^{k-t_k}x_{t_k}$ is actually a linear combination of the process noise signals from time $t_k+1$ up to $k$
	\begin{equation*}
	x_k-A^{k-t_k}x_{t_k}=\sum_{j=t_k+1}^{k}A^{k-j}w_j
	\end{equation*}
	as follows from the system dynamics \eqref{EQN_system}. If the critical event occurs at some time $k_0$, then the eavesdropper permanently loses all information about the process noise $w_{k_0}$ at time $k_0$. This loss of information is amplified by the unstable system dynamics over time leading to the eavesdropper's unbounded error. \hfill $\diamond$
\end{remark} 
\begin{remark}\label{REM_channel}
Suppose that the channel outcomes are independent over time, and suppose that there is a positive probability that the critical event occurs at any time $k$, i.e., $P\paren{\gamma_{u,k}=1,\,\gamma_{k}=0}>\delta>0$. For example in a wireless communication setting this may be due to attenuation of the transmitted signal at the eavesdropper or due to environmental interference. Then condition \eqref{EQN_theorem_condition} for perfect secrecy by Theorem \ref{THM_perfect_secrecy} follows from the Borel-Cantelli lemma \cite{durrett2010probability}. 
\hfill $\diamond$
\end{remark}

\subsection{Rate of increase of eavesdropper's error covariance}
Another property of the proposed coding scheme is that the rate of increase of the eavesdropper's error covariance is asymptotically optimal.  Once the critical event occurs, the eavesdropper's error behaves asymptotically as the open loop prediction error, i.e. when all measurements are lost, which is the largest possible error for the eavesdropper. 

Formally, let us denote the open-loop prediction estimate and error covariance matrix by:
\begin{equation}\label{EQN_Open_Loop_Estimation}
x^{op}_k=\expe{ x_{k}},\quad P^{op}_k=\var{x_k},
\end{equation}
which implies: 
\begin{equation}\label{EQN_Open_loop_covariance}
P^{op}_k=AP^{op}_{k-1}A'+Q,
\end{equation}
with $P^{op}_0=\Sigma_0$. 
The expected eavesdropper's error is always smaller than the open-loop prediction error  since it holds that:
\begin{equation}\label{EQN_Open_Loop_Optimal}
P^{op}_{k}=\var{x_k}\succeq \expe{\cov{x_k|\info{k}}}=\expe{P_k}
\end{equation}
for any information set $\info{k}$, where the inequality follows by \cite[p. 230]{durrett2010probability}. Thus, the open-loop prediction error is the maximum possible in expectation.

The following result, which is a corollary of Theorem \ref{THM_perfect_secrecy}, shows that the open-loop prediction error covariance is upper bounded by a multiple of the eavesdropper's error, where the multiplicative constant depends on the time $k_0$ of the first critical event. This implies that the rate of increase of our coding scheme is asymptotically optimal with respect to the eavesdropper's error once the critical event occurs.
\begin{corollary}[Rate of increase]\label{THM_rate_of_increase}
Consider system \eqref{EQN_system}, with channel model \eqref{EQN_channel_model} and coding scheme \eqref{EQN_Coding}. Let $k_0$ be the first time the critical event \eqref{EQN_critical_event} occurs. Let also $P^{op}_k$, $P_k$ be the open-loop prediction error covariance \eqref{EQN_Open_loop_covariance} and estimation error covariance \eqref{EQN_estimate_variance_eavesdropper} matrices respectively. Then,	
\begin{equation}\label{EQN_rate_of_increase}
\Tr P^{op}_k \le c \rho(A)^{-2k_0} \paren{\Tr P_k+1},
\end{equation}
where $c>0$ is some constant independent of $k_0$ \hfill $\diamond$
\end{corollary}

The determining factor in the above inequality is the time $k_0$ of the first critical event. If the critical event occurs more than once, although the eavesdropper's error gets larger, its rate of growth does not differ from the case when it occurs just once. If the probability of the critical event occurring is one, then also $P\paren{\rho(A)^{-2k_0}>0}=1$. Hence, with even a single occurrence, the error starts increasing as in the open-loop asymptotically, which is a very strong guarantee.

One caveat is that the first time $k_0$ the critical event occurs is in general random and not in our control. 
If the eavesdropper's interception rate is very high, the event may take some time to occur. Then, the term $\rho(A)^{-2k_0}$ is smaller, which means that secrecy is compromised at the first time steps. 
\begin{remark}
A possible remedy to accelerate the critical event, is to force it by using a complementary and perhaps more expensive encoding, e.g., encryption. In this case, it is sufficient to securely and reliably transmit just the first packet at time $k=0$ using the complementary scheme. Then, letting our cheap coding scheme take over achieves perfect secrecy. In that sense, State-Secrecy Codes can be effectively used alongside other coding schemes. Theorem \ref{THM_perfect_secrecy} allows us to concentrate our more expensive defense efforts in a single transmission and still achieve perfect secrecy.
\end{remark}

\section{Perfect secrecy with output measurements}\label{Section_Output}
In this section, we adapt the coding scheme \eqref{EQN_Coding} to achieve perfect secrecy in the case of output measurements (general $C$, $R\succ 0$). This is achieved at the expense of some additional computational cost at the sensor.
\begin{figure}[t] \centering{
		\includegraphics[scale=0.7]{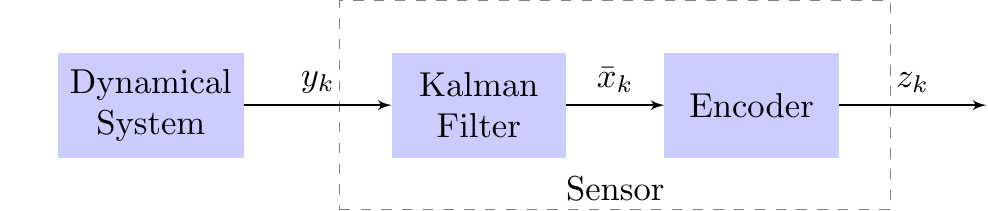}  
		\caption {In the case of output measurements, the sensor locally implements a Kalman Filter and computes the estimates $\bar{x}_k$. Then, it forms the encoded signal $z_k$ based on these local estimates $\bar{x}_k$.}
		\label{Figure_Output_Coding} }
\end{figure}
Before transmitting and before applying any coding, we propose that the sensor implements a local Kalman filter scheme (see also \cite{hespanha2007survey})   as shown in Figure \ref{Figure_Output_Coding} and computes the local estimates:
\begin{equation}\label{EQN_local_Kalman_estimate}
\bar{x}_k=\expe{x_k|\vct{y}_{0:k}},\,\bar{x}_{-1}=\expe{x_0}=0
\end{equation}
with prediction error covariance matrix:
\begin{equation}\label{EQN_local_Kalman_covariance}
\bar{P}_{k+1|k}=\var{x_{k+1}|\vct{y}_{0:k}},\,\bar{P}_{0|-1}=\var{x_{0}}=\Sigma_0
\end{equation}
where $\vct{y}_{0:k}=\paren{y_0,\dots,y_k}$. Next, the sensor applies a State-Secrecy Code to the local state estimates: 
\begin{equation}\label{EQN_ouput_coding}
\y{k}=\bar{x}_k-A^{k-t_k}\bar{x}_{t_k},
\end{equation} 
similar to the code in \eqref{EQN_Coding} with $t_k$ being the reference time defined in \eqref{EQN_reference_time}. 

Now, recall that in the classic Kalman filter derivation \cite{anderson2005optimal}, we have the following recursive equation:
\begin{equation}\label{EQN_system_local_Kalman}
\bar{x}_k=A\bar{x}_{k-1}+K_k\paren{y_k-CA\bar{x}_{k-1}}=A\bar{x}_{k-1}+K_k\bar{w}_k,
\end{equation}
where we use $\bar{w}_k$ to denote the innovation sequence $y_k-CA\bar{x}_{k-1}$ and \[
K_k=\bar{P}_{k|k-1} C'\paren{C\bar{P}_{k|k-1}C'+R}^{-1}\]
is the Kalman gain. 
It follows that the innovation sequence $y_k-CA\bar{x}_{k-1}$ is Gaussian white noise \cite{anderson2005optimal} with covariance $C\bar{P}_{k|k-1}C'+R$. In this sense, equation \eqref{EQN_system_local_Kalman} is similar to the state equation \eqref{EQN_system}. 

Since the pair $\paren{A,C}$ is detectable and $Q\succ 0$ (so that $\paren{A,Q^{1/2}}$ is controllable), the prediction matrix $\bar{P}_{k|k-1}$ of the Kalman filter converges to a limit $\bar{P}$, which is the positive semidefinite solution of the discrete algebraic Riccati equation:
\begin{equation}\label{EQN_DARE}
\bar{P}=A\bar{P}A'+Q-A\bar{P}C'\paren{C\bar{P}C'+R}^{-1} C\bar{P}A'
\end{equation}
while the Kalman gain converges to \begin{equation}\label{EQN_Kalman_gain}
K=\bar{P} C'\paren{C\bar{P}C'+R}^{-1}.
\end{equation}
We assume that at time $k=0$ the sensor has an initial local state estimate with covariance equal to the steady state error covariance $\Sigma_{0}=\bar{P}$. Since the Kalman filter converges fast, it is reasonable to assume that it has already converged at the beginning of the system operation.

The main result of this section is presented in the following theorem. It states that coding scheme \eqref{EQN_ouput_coding} achieves the same perfect secrecy guarantees as in the case of direct state information case in Section \ref{Section_Perfect}. The user's estimation error is optimal, while just a single occurrence of the critical event renders the eavesdropper's error unbounded with exponential rate of increase. 
\begin{theorem}\label{THM_output}
	Consider system \eqref{EQN_system} with output model \eqref{EQN_output}, channel model \eqref{EQN_channel_model}, coding scheme \eqref{EQN_ouput_coding} and $\Sigma_0=\bar{P}$. If
	\begin{equation}\label{EQN_output_hypothesis}
	\mathbb{P}( \gamma_{u,k}=1,\,\gamma_{k}=0,\text{ for some }k\ge 0 )=1, 
	\end{equation}
	then:
	\renewcommand{\labelenumi}{(\roman{enumi})}
	\begin{enumerate}
		\item perfect secrecy is achieved according to Definition \ref{DEF_perfect_secrecy}.
		\item conditioned on the event $\set{\gamma_{u,k_0}=1,\,\gamma_{k_0}=0}$ for some $k_0\ge 0$, the eavesdropper's estimation error grows unbounded for $k \geq k_0$ as 
		\begin{equation}\label{EQN_output_main_result_asymptotics}
		\Tr{P_k}\ge c \rho\paren{A}^{2\paren{k-k_0}}-c',
		\end{equation}
	\end{enumerate}
	where $P_k$ is the error covariance defined in \eqref{EQN_estimate_variance_eavesdropper}  and $c,c'>0$ are some constants independent of $k_0$.
	\hfill $\diamond$
\end{theorem} 

To prove Theorem \ref{THM_output},  we can use the techniques of the previous section to show an intermediate result first; that the eavesdropper's error with respect to $\bar{x}_k$ grows unbounded when the critical event occurs, as the following lemma states. We denote the mean square estimate of  $\bar{x}_k$ and the corresponding conditional error covariance matrix by: 
\begin{equation}\label{EQN_Secondary_Estimation}
\est_k=\expe{ \bar{x}_{k}\vert \info k},\quad \estm_k=\var{\bar{x}_{k}\vert\info k},
\end{equation}
with $\bar{\Sigma}_0=\var{\bar{x}_{0}}=K\paren{C\bar{P}C'+R}K'$.

\begin{lemma}\label{THM_output_lemma}
	Consider system \eqref{EQN_system} with output model \eqref{EQN_output}, channel model \eqref{EQN_channel_model}, coding scheme \eqref{EQN_ouput_coding} and $\Sigma_0=\bar{P}$. Conditioned on the event $\event=\set{\gamma_{u,k_0}=1,\,\gamma_{k_0}=0}$ for some $k_0\ge 0$, the eavesdropper's estimation error with respect to $\bar{x}_k$ grows unbounded for $k \geq k_0$ as 
		\begin{equation}\label{EQN_rate_lemma_output_secondary}
		\Tr{\estm_k}\ge c \rho\paren{A}^{2\paren{k-k_0}}
		\end{equation}
where $\estm_k$ is the error covariance defined in \eqref{EQN_Secondary_Estimation} and $c>0$ is a constant independent of $k_0$. 
\hfill $\diamond$
\end{lemma}
We use this intermediate result to show that the eavesdropper's error with respect to $x_k$ also grows unbounded when the critical event occurs. The main idea is to lower-bound the error $\Tr P_k$ in terms of $\Tr \estm_k$. The proof is included in the Appendix.

\section{Simulations}\label{Section_Simulations}
In this section we illustrate the efficiency of our proposed coding schemes in numerical simulations. We consider two scenarios. In the first one, we contrast the performance achieved by the State-Secrecy Codes with the one achieved by the mechanisms in \cite{tsiamis2016state}, \cite{leong2017remote}. The comparison is made assuming direct state measurements. In the second one, we compare the user's and eavesdropper's estimation performance in the case of output measurements. The system under consideration has state matrix $A=\matr{{cc}1.2& 0.1\\0& 0.5}$ and process noise covariance matrix $Q=\matr{{cc}0.6& 0.2\\0.2& 0.5}$. For the channel model, we assume that the channel outcomes are independent across time and stationary with probabilities $P\paren{\gamma_{u,k}=i,\gamma_{k}=j}=p_{ij}$, for $i,j\in\set{0,1}$.
For the estimation scheme of the eavesdropper we used equation \eqref{EQN_fundamental_conditional_property} in Appendix. Since the user can decode all signals, we used the formula:
\begin{equation*}
P_{u,k}=
\left\{\begin{aligned}
&\bar{P}-KC\bar{P}&&\text{ if }\gamma_{u,k}=1\\
&AP_{u,k-1}A'+Q&&\text{ if }\gamma_{u,k}=0
\end{aligned}\right.
\end{equation*}
where $\bar{P}-KC\bar{P}=0$ in the case of direct state measurements.

\begin{figure}[t] \centering{
		\includegraphics[scale=0.48]{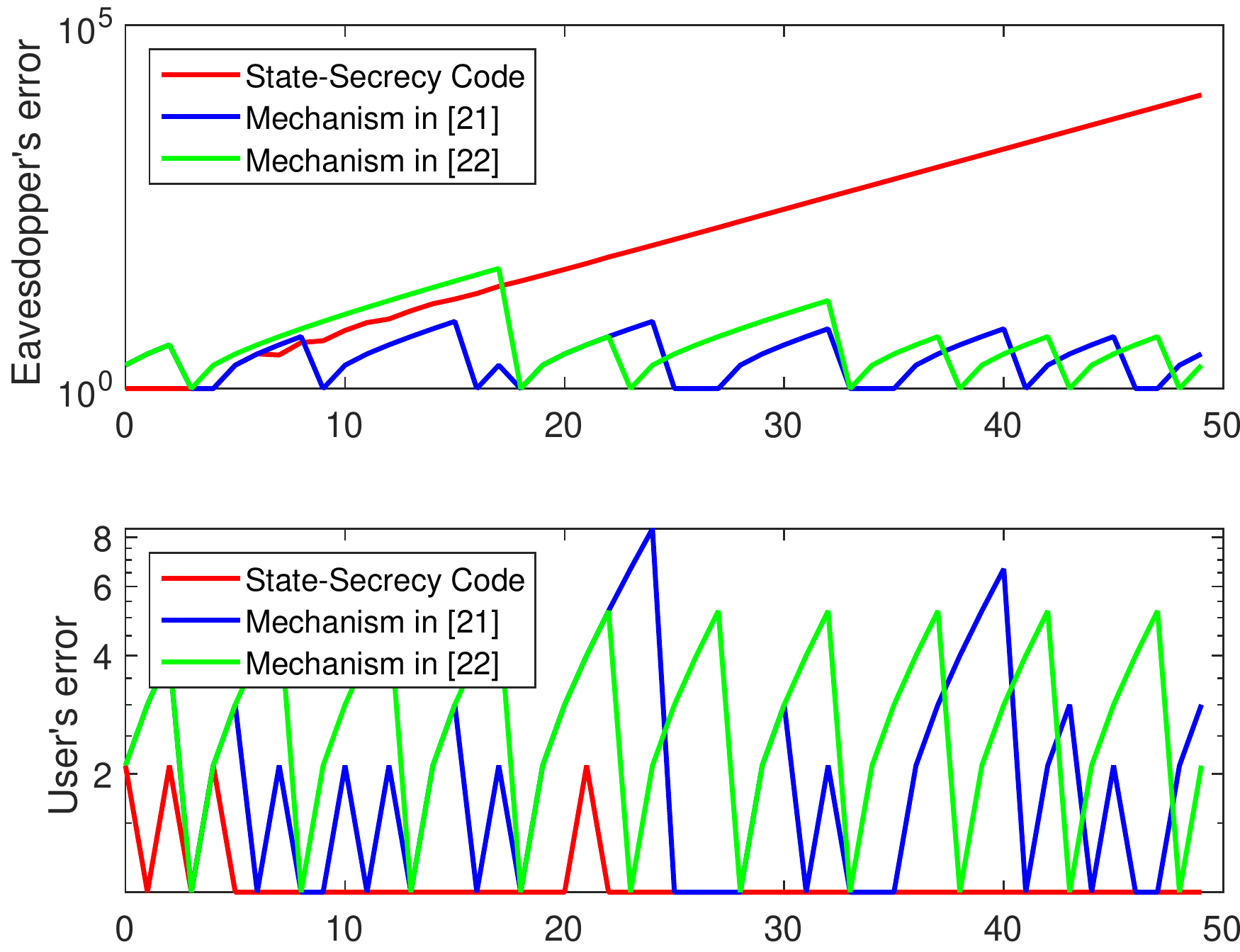}}   
	\caption {We compare our proposed coding scheme with the mechanisms in \cite{tsiamis2016state}, \cite{leong2017remote} for a typical channel outcome sequence in the case of direct state measurements. For the log-plots, we use function $\log(x+1)$ instead of $\log(x)$. We observe that it significantly outperforms the mechanisms in both confidentiality (eavesdropper's error) and efficiency (user's error).}
	\label{Figure_compare_secrecy}
\end{figure}
In the first scenario, with direct state information ($C=I$, $R=0$), we assume that the channel outcomes have the probabilities $p_{11}=0.54$, $p_{00}=0.04$, $p_{01}=0.06$ and $p_{10}=0.36$ and initial state error covariance matrix $\Sigma_0=Q$. We compare the performance of our code with the one of the mechanisms in \cite{tsiamis2016state} and \cite{leong2017remote}--see Figure~\ref{Figure_compare_secrecy}. The comparison is made for the same sequence of channel outcomes, with respect to the user's and eavesdropper's estimation errors ($\Tr\paren{P_{u,k}}$ and $\Tr\paren{P_{k}}$ respectively). For the mechanism in \cite{tsiamis2016state}, which randomly withholds state information with probability $p$, we selected  $p=0.29$. For the infinite horizon mechanism in \cite{leong2017remote}, which transmits state information only if the user loses more than $s$ consecutive packets, we used $s=5$. 
Notice that the eavesdropper's error is small very often under the mechanisms in \cite{tsiamis2016state} and \cite{leong2017remote}, since unboundedness is guaranteed in expectation, not almost surely. In contrast, our coding scheme achieves  unbounded eavesdropper's error for every channel sequence with probability one. Also notice that the user's estimation performance is degraded in \cite{tsiamis2016state}, \cite{leong2017remote}.

\begin{figure}[t] \centering{
		\includegraphics[scale=0.48]{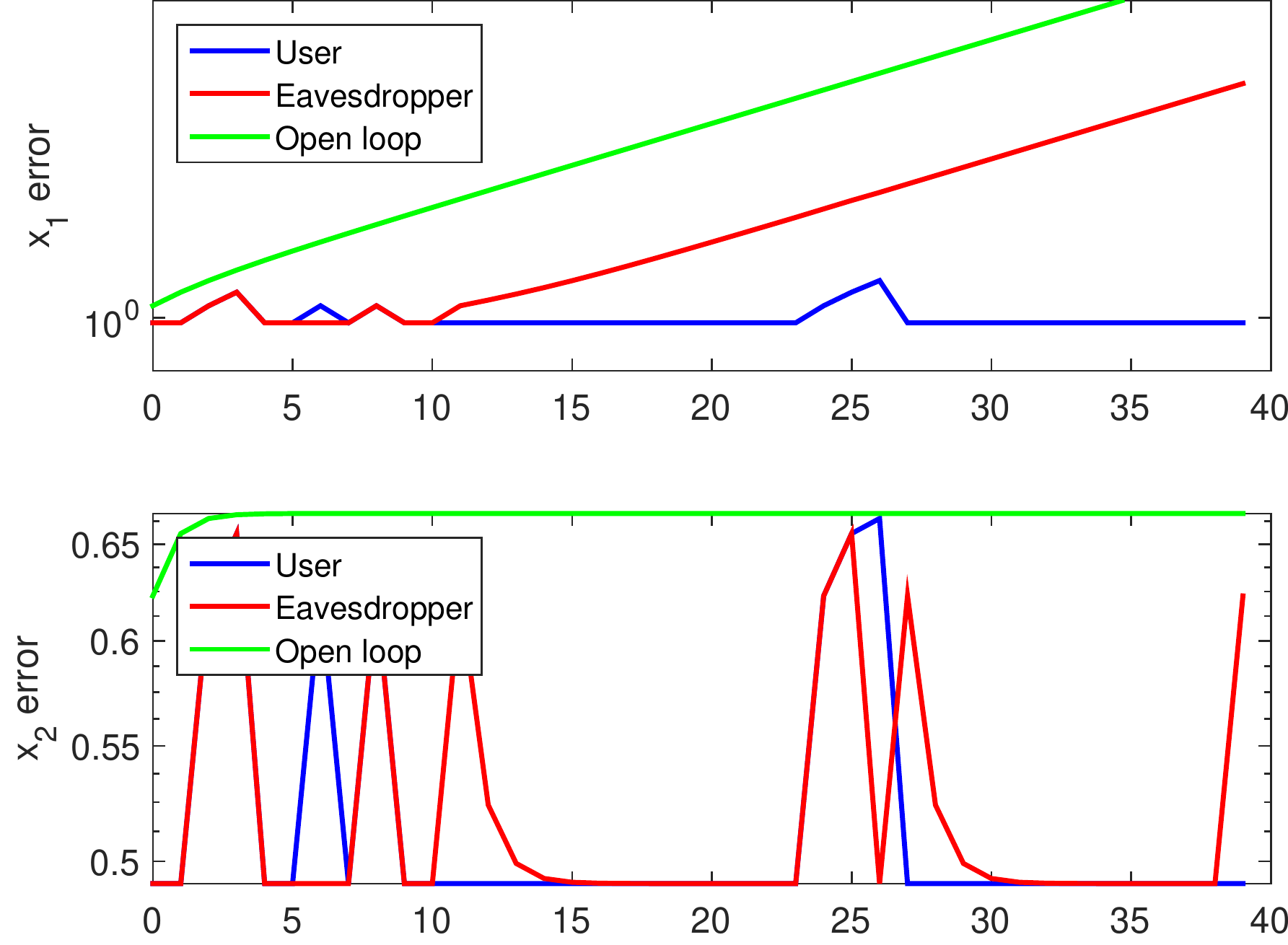}}   
	\caption {We compare the eavesdropper's, user's and open loop errors for the states $x_1$ and $x_2$ in the case of output measurements. Notice that the critical event occurs at time $k=11$. The eavesdropper's error regarding the unstable part grows unbounded with asymptotically the same rate as the open-loop error. The eavesdropper still has knowledge about the stable state $x_2$. However, the open-loop prediction error for $x_2$ shows that the eavesdropper has bounded error with respect to the stable dynamics, regardless the code. 
	}
	\label{Figure_Individual_State}
\end{figure}
For the second scenario with output state measurements, we assume that $C=\matr{{cc}1 & 1}$, $R=1$  and the channel outcomes have the probabilities $p_{11}=0.7$, $p_{00}=0.1$, $p_{01}=0.1$ and $p_{10}=0.1$. We also assume $\Sigma_0=\bar{P}$ (see \eqref{EQN_DARE}). In Figure \ref{Figure_Individual_State}, we plot the user's and eavesdropper's estimation errors over time for the states $x_1$ 
and $x_2$, 
i.e. the two diagonal elements of the covariance matrices respectively. 
We also plot the open-loop prediction error (see \eqref{EQN_Open_loop_covariance}), to compare it with the eavesdropper's error. 
As shown in Figure \ref{Figure_Individual_State}, the eavesdropper's error for the unstable state $x_1$ starts growing unbounded at time $k=11$, when the first critical event occurs; the rate of increase is the asymptotically the same as in the open-loop case. It is worth noting that the eavesdropper's error for the stable state $x_2$ remains bounded. 
But it is fundamentally impossible to have unbounded error for the stable part, regardless the code, as even the open-loop prediction error is bounded for state $x_2$.

\section{Conclusion}\label{Section_Conclusion}
The presence of an eavesdropper adds new challenges to the problem of remote estimation. Nonetheless, by using a simple State-Secrecy Code, based on acknowledgment signals from the user back to the sensor, we can achieve powerful confidentiality guarantees with minimal computational cost. At the same time the user's estimation performance is optimal. By exploiting the random packet erasures, the process noise, and the unstable dynamics, perfect secrecy is achieved with just a single occurrence of the critical event, when the user receives more information than the eavesdropper.

Future work includes an implementation and experimental evaluation of the proposed scheme. We also seek to compare the computational load of encryption with that of our code.
Another open question is how to adapt State-Secrecy Codes to offer more confidentiality guarantees for the stable part of the state. Finally, adapting State-Secrecy Codes to closed-loop control systems is an interesting future direction. 

\appendices
\appendices

\section{Estimation formulas}
In this appendix, we derive two formulas for the eavesdropper's estimation error covariance in the case of state measurements. The first one, equation \eqref{EQN_fundamental_conditional_property} in Lemma \ref{THM_conditional_expectation}, expresses the conditional expectation with respect to the non-Gaussian  eavesdropper's information $\info{k}$ (see \eqref{EQN_eavesdropper_information}) in terms of conditional expectations with respect to Gaussian variables. The second one, equation \eqref{EQN_Conditional_Recursive_formula} in Proposition \ref{THM_recursive}, is the main estimation formula, which will be used to prove Lemma \ref{THM_worst_case}. An identical analysis can be followed to find the estimation formulas for $\est_k,$ $\estm_k$ in \eqref{EQN_Secondary_Estimation} for the case of output measurements. Thus, it is omitted.

We denote by $\x{k}=\matr{{ccc}x_k',\dots,x_0'}'$ the batch state vector up to time $k$ taking values in $\REAL^{(k+1)n}$. 
We also denote the batch channel outcomes up to time $k$ by  
\begin{equation}\label{EQN_batch_channel_outcomes}
\g{k}=\paren{\vct{\gamma}_{u,0:k},\vct{\gamma}_{0:k}}=\paren{\gamma_{u,0},\dots,\gamma_{u,k},\gamma_{0},\dots,\gamma_{k}}
\end{equation}
taking values in the set $\set{0,1}^{2k+2}$. 
Now let \[s_{0:k}=\paren{s_{u,0},\dots,s_{u,k},s_{0},\dots,s_{k}}\] be any fixed element of $\set{0,1}^{2k+2}$. For symbol economy, we write $s=s_{0:k}$ and we omit the subscript $0:k$. Given this element $s$, we define  $s_{0:m}=\paren{s_{u,0},\dots,s_{u,m},s_{0},\dots,s_{m}}$ to be the part of $s$ until time $m$, for $m\le k$. We also define $s_{u,0:m}=\paren{s_{u,0},\dots,s_{u,m}}$ to be the part of $s$ related to the user until time $m$, for $m\le k$.

 Next, we alternatively describe the information $\mathcal{I}_k$  that the eavesdropper has on the event that the channel outcomes $\g{k}$ take exactly the value $s$. Formally it is the set
\begin{equation}\label{EQN_induced_info}
\begin{aligned}
&\ifo{s}=\set{z_m\paren{s}:s_{m}=1,\,m\le k},\,\text{where}\\
&z_m\paren{s}=x_m-A^{m-t_m\paren{s_{u,0:m-1}}}x_{t_m\paren{s_{u,0:m-1}}}
\end{aligned}
\end{equation}
and $t_m\paren{\cdot}$ is defined in \eqref{EQN_reference_time}. In other words, given fixed channel outcomes, the eavesdropper only keeps the successfully intercepted channel outputs from $\vct{h}_{0:k}$ (see \eqref{EQN_channel_model}). Notice that every element of $\ifo{s}$ depends linearly on $\x{k}$.  The next example clarifies this definition. 
\begin{example}
Suppose we are given $s=s_{0:2}=\paren{s_{u,0},s_{u,1},s_{u,2},s_{0},s_{1},s_{2}}=\paren{1,0,1,1,0,1}$. Then, $t_0(s)=-1$, $t_1(s)=t_2(s)=0$
and  $$\mathcal{J}_2\paren{s}=\set{x_2-A^2x_0,x_0}.$$
Notice that $\mathcal{J}_2\paren{s}$ depends only on $\x{2}$ and is Gaussian.
\hfill $\diamond$
\end{example}

\subsection{Sigma algebra of $\info{k}$.}
Before we prove the estimation formulas, we describe the sigma-algebra $\sigma\paren{\info{k}}$, induced by the eavesdropper's information $\info{k}$ (see \eqref{EQN_eavesdropper_information}). 
The following lemma describes the sigma algebra $\sigma\paren{\info{k}}$, in terms of $\sigma\paren{\ifo{s}}$. 
\begin{lemma}[Sigma algebra of $\info{k}$]\label{LEM_sigma_algebra}
Fix a $k$ and let $S=\set{0,1}^{2k+2}$. Let $\info{k}$, $\g{k}$, and $\ifo{s}$ be as defined in \eqref{EQN_eavesdropper_information}, \eqref{EQN_batch_channel_outcomes}, and \eqref{EQN_induced_info}. Then every set $D$ in the sigma algebra of $\info{k}$  has the following form:
\begin{equation}\label{EQN_sigma_algebra_form}
D=\bigcup_{s\in S}\set{\g{k}=s,\ifo{s}\in F_s},
\end{equation}
for some $F_s\in \sigma\paren{\ifo{s}}$, for all $s\in S$. 
\end{lemma}
\begin{proof}
Since $\g{k}$ is a discrete-valued random vector, we can partition every set $D\in \sigma\paren{\info{k}}$ according to the values of $\g{k}$ as:
\[D=\bigcup_{s\in S} D_s,\]
where \[D_s=\set{\g{k}=s,\vct{h}_{0:k}\in E_s},\]
for some set $E_s$ that belongs in $\sigma\paren{\vct{h}_{0:k}}$. But we have that:
\[\set{\gamma_t=0}\Leftrightarrow \set{h_t=\varepsilon},\text{ for all }0\le t\le k. \]
Thus, if $\gamma_t=0$ in in $D_s$, this fully describes $h_t$ in $D_s$.
As a result, the set $D_s$ can be equivalently described by
\[D_s=\set{\g{k}=s,\ifo{s}\in F_s},\]
for some $F_s\in \sigma\paren{\ifo{s}}$.
\end{proof}

\subsection{First estimation formula}
Here, we express the conditional expectation with respect to $\info{k}$ in terms of conditional expectations with respect to $\ifo{s}$. This is presented in the following lemma, and is a consequence of Lemma \ref{LEM_sigma_algebra} and independence of $\x{k}$ and $\g{k}$.
 \begin{lemma}\label{THM_conditional_expectation}
 Fix a $k$ and let $\info{k}$, $\g{k}$, and $\ifo{s}$ be as defined in \eqref{EQN_eavesdropper_information}, \eqref{EQN_batch_channel_outcomes}, and \eqref{EQN_induced_info}. Suppose $Y=Y\paren{\g{k}}$ is an integrable random vector, such that:
 \begin{equation*}
 Y\paren{\g{k}}=\sum_{s\in S}Y\paren{s}\1_{\set{\g{k}=s}}
 \end{equation*}
 where $Y\paren{s}$ depends only on $\x{m}$, for some $m$ (possibly different than $k$). Then:
 \begin{equation} \label{EQN_fundamental_conditional_property}
 \expe{Y|\info{k}}=\expe{Y\paren{s}|\ifo{s}}\text{ in }\set{\g{k}=s}.
 \end{equation}
  By $\1$ we denote the indicator function.\hfill $\diamond$
 \end{lemma}
 \begin{proof}
 	Define $S=\set{0,1}^{2k+2}$.
Notice that $\g{k}$ is $\sigma\paren{\info{k}}$-measurable. Thus, we have:
 	\begin{align*}
 	\expe{Y|\info{k}}&= \sum_{s\in S}\1_{\set{\g{k}=s}}\expe{Y(s)|\info{k}}\\
 	&=\expe{Y(s)|\info{k}}\text{ in }\set{\g{k}=s}.
 	\end{align*}
 	Thus, it is sufficient to show that \begin{equation}\label{EQN_Conditional_lemma_goal}
 	\1_{\set{\g{k}=s}}\expe{Y(s)|\info{k}}=\1_{\set{\g{k}=s}}\expe{Y(s)|\ifo{s}}.
 	\end{equation}
 	Since $\1_{\set{\g{k}=s}}\expe{Y(s)|\info{k}}=\expe{\1_{\set{\g{k}=s}} Y(s)|\info{k}}$, the truth of \eqref{EQN_Conditional_lemma_goal} can be verified if we show the basic property of conditional expectation: \[\expe{\1_{D}\1_{\set{\g{k}=s}}Y(s)}=\expe{\1_{D}\1_{\set{\g{k}=s}}\expe{Y(s)|\ifo{s}}},\] for any $D\in \sigma\paren{\info{k}}$. 
 	From Lemma \ref{LEM_sigma_algebra}, we have:
 	  \[D=\bigcup_{s\in S}D_s,\]
 	where
 	\[D_s=\set{\g{k}=s,\ifo{s}\in F_s},\] for some $F_s\in\sigma\paren{ \ifo{s}}$. As a result, $\1_{\set{\g{k}=s}}\1_{D}=\1_{\set{\g{k}=s}}\1_{\ifo{s}\in F_s}$. 
 	
 	Observe that  $\ifo{s}$ and $Y(s)$ depend only on the values of $x_t$, $t\le m$, while the indicator $\1_{\set{\g{k}=s}}$ depends on the channel outcomes. Hence:
 	\begin{align*}
 	&\expe{\1_{D}\1_{\set{\g{k}=s}}Y(s)}= \expe{\1_{\set{\g{k}=s}}\1_{\ifo{s}\in F_s}Y(s)}\\
 	&=\expe{\1_{\set{\g{k}=s}}}\expe{\1_{\ifo{s}\in F_s}Y(s)}\\
 	&=\expe{\1_{\set{\g{k}=s}}}\expe{\1_{\ifo{s}\in F_s}\expe{Y(s)|\ifo{s}}}\\
 	&=\expe{\1_{\set{\g{k}=s}}\1_{\ifo{s}\in F_s}\expe{Y(s)|\ifo{s}}}\\
 	&=\expe{\1_{D}\1_{\set{\g{k}=s}}\expe{Y(s)|\ifo{s}}}
 	\end{align*}
 	where the second and fourth equalities follow from independence, while the third follows from the properties of conditional expectation.
 \end{proof}
 The benefit of Equation \eqref{EQN_fundamental_conditional_property} is that it allows us to work with $\ifo{s}$, which includes only Gaussian elements, that depend linearly on $\x{k}$. In contrast, it is not easy to work directly with $\info{k}$, which is non-Gaussian.
\subsection{Second estimation formula}
 We finally prove our main estimation formula for $P_k$ by leveraging the technical intermediate formula \eqref{EQN_fundamental_conditional_property}.
\begin{proposition}[Estimation formula]\label{THM_recursive}
	Consider system \eqref{EQN_system}, channel \eqref{EQN_channel_model} under coding scheme \eqref{EQN_Coding}. Fix any $k \ge  0$. Let the covariance matrix of $x_k$ and $\y{k}$ given $\info{{k-1}}$ be written in a block form: $$\var{\matr{{c}x_{k}\\\y{k}}|\info{k-1}}=\matr{{cc}\Sigma_{xx} & \Sigma_{xz}\\\Sigma_{zx}&\Sigma_{zz}}$$
	Then, $\Sigma_{xx}=AP_{k-1}A'+Q$ if $k>0$ and $\Sigma_{xx}=\Sigma_0$ if $k=0$, where $P_{k-1}$ is the estimation error covariance of the eavesdropper at time $k-1$ defined in \eqref{EQN_estimate_variance_eavesdropper}, and the conditional covariance at time $k$ is given by:
	\begin{equation}	\label{EQN_Conditional_Recursive_formula}
	P_{k}=\Sigma_{xx}-\gamma_{2,k}\Sigma_{xz}\paren{\Sigma_{zz}}^{\dagger}\Sigma_{zx},
	\end{equation}
		where $(\cdot)^{\dagger}$ denotes the Moore-–Penrose pseudoinverse \cite{horn2012matrix}.
	\hfill $\diamond$
\end{proposition}
\begin{proof}
	First, we prove the formula for $\Sigma_{xx}$.  When $k=0$, $\info{-1}$ is an empty set, thus, $\Sigma_{xx}=\Sigma_0$.
	When $k>0$, we argue that $\Sigma_{xx} = AP_{k-1}A'+Q$. By the system dynamics in \eqref{EQN_system} we have that $x_k = A x_{k-1} + w_k$. But by assumption the noise $w_k$ is independent of the state vectors $x_0, \ldots, x_{k-1}$ and of $\g{k}$ and, hence, also independent of $\info{k-1}$. Thus,  the prediction estimation is $\expe{x_k|\info{k-1}}=A \hat{x}_{k-1}$ and the covariance of $x_k-\expe{x_k|\info{k-1}}$ given $\info{k-1}$ equals $\Sigma_{xx}=AP_{k-1}A'+Q$.

	Second, we prove \eqref{EQN_Conditional_Recursive_formula}. Fix an element $s=s_{0:k} \in \set{0,1}^{2k+2}$, and consider the event $C=\set{\g{k}=s_{0:k}}$. 
	By the definition \eqref{EQN_estimate_variance_eavesdropper} of $P_k$ and equation \eqref{EQN_fundamental_conditional_property}, we have $P_{k}=\var{x_k|\ifo{s_{0:k}}}$ in $C$.
	There exist two cases: 
	
	\noindent \textit{Case I:} $s_{k}=0$. In this case, by the definition \eqref{EQN_induced_info}, we have $\ifo{s_{0:k}}=\ifo{s_{0:k-1}}$. Therefore:
\begin{align*}
P_{k}&=\var{x_k|\ifo{s_{0:k}}}=\var{x_k|\ifo{s_{0:k-1}}}\\&=\var{x_k|\info{k-1}}=\Sigma_{xx}\text{ in }C.
\end{align*}
where the third equality follows from \eqref{EQN_fundamental_conditional_property}.
	
	\noindent \textit{Case II:} $s_{k}=1$. Here, by \eqref{EQN_fundamental_conditional_property}: \begin{align*}P_k& = \var{x_k|\ifo{s_{0:k}}}\\
	& = \var{x_k|\ifo{s_{0:k-1}}, z_k\paren{s_{0:k}}},\,\text{in }C,
	\end{align*}
	 where $z_k\paren{s_{0:k}}$ is defined in \eqref{EQN_induced_info}. Since all variables are Gaussian, we can compute the posterior estimation error using the Schur complement formula applied to the covariance matrix of $x_k$ and $z_k\paren{s_{0:k}}$ given $\ifo{s_{0:k-1}}$ (see \cite{anderson2005optimal}).
	By \eqref{EQN_fundamental_conditional_property}, in $C$ we have 
	\begin{align*}&\var{\matr{{c}x_{k}\\ z_k\paren{s_{0:k}}}|\ifo{s_{0:k-1}}}=\matr{{cc}\Sigma_{xx} & \Sigma_{xz}\\\Sigma_{zx}&\Sigma_{zz}}.
	\end{align*}
	Hence, by the Schur's complement formula, we have $P_{k}=\Sigma_{xx}-\Sigma_{xz}\paren{\Sigma_{zz}}^{\dagger}\Sigma_{zx} $ in $C$, when $s_{k}=1$.

	We can describe both cases with one equation:
	\begin{align*}	
	P_{k}&=\Sigma_{xx}-s_{k}\Sigma_{xz}\paren{\Sigma_{zz}}^{\dagger}\Sigma_{zx}\text{ in }C\\
	&=\Sigma_{xx}-\gamma_{k}\Sigma_{xz}\paren{\Sigma_{zz}}^{\dagger}\Sigma_{zx}\text{ in }C,
	\end{align*}
	since $s_{k}=\gamma_{k}$ in $C$. But this holds for any event $C=\set{\g{k}=s_{0:k}}$. Thus, the result holds everywhere, since the events $\set{\g{k}=s}$, $s\in \set{0,1}^{2k+1}$ are a partition of the probability space.
\end{proof}
\section{Proofs of results}
\subsection{Proof of Lemma \ref{THM_worst_case}}
We will use induction to prove the formula \eqref{EQN_Lemma_goal_equation} of the lemma.
For $k=k_0$ it is immediate. Suppose it is true for $k-1\ge k_0$. Since in $\EV$ the eavesdropper receives all packets for $k>k_0$, we have $
\gamma_{k}=1$, for $k>k_0$. Hence, according to the recursive formula \eqref{EQN_Conditional_Recursive_formula}, to find $P_{k}$, we should compute the covariance of $x_{k}$ and $\y{k}$, conditioned on $\info{k-1}$. 
By independence of $w_k$ from $\info{k-1}$, it follows that $x_{k}-\expe{x_{k}|\info{k-1}}=A\paren{x_{k-1}-\hat{x}_{k-1}}+w_{k}$.
Now, we claim $\y{k}-\expe{\y{k}|\info{k-1}}=w_k$, for $k>k_0$, which we prove in the next paragraph. Thus, the covariance matrix of $x_{k}$ and $\y{k}$, conditioned on $\info{k-1}$ is:
\begin{align*}&\var{\matr{{c}x_{k}\\\y{k}}|\info{k-1}}=\matr{{cc}AP_{k-1}A'+Q& Q\\Q & Q},\text{ in }\event\cap\EV 
\end{align*}
Thus, by \eqref{EQN_Conditional_Recursive_formula}, for $k>k_0$ we have
\begin{align}
P_{k}&=AP_{k-1}A'+Q-QQ^{-1}Q
=AP_{k-1}A'
\end{align}
in $\event\cap\EV$. Hence, by the induction hypothesis, we get  \eqref{EQN_Lemma_goal_equation}.

Finally, we prove the claim 
\begin{equation}\label{EQN_lemma_worst_case_intermediate}
\y{k}-\expe{\y{k}|\info{k-1}}=w_k,\,\text{for }k>k_0,\text{ in }\event\cap \EV.
\end{equation}
Since the critical event 
happened at time $k_0$, the reference time at $k_0+1$  is $t_{k_{0}+1}=k_{0}$  in $\event$ (equation \eqref{EQN_reference_time}). Hence, all reference times $t_k$ , for $k>k_0$ satisfy $t_k\ge k_0$ in $\event$. There are only two possible cases depending on $\gamma_{u,k-1}$: \\
	\noindent \textit{Case I:} $t_k=k-1\ge k_0$, when $\gamma_{u,k-1}=1$\\
	\noindent \textit{Case II:} $t_k=t_{k-1}\ge k_0$ when $\gamma_{u,k-1}=0$\\
In the former one, the intercepted signal by \eqref{EQN_Coding} is $\y{k}=x_k-Ax_{k-1}=w_{k}$. But the process noise $w_k$ is independent of $\info{k-1}$, thus, $\expe{w_k|\info{k-1}}=0$ and equation \eqref{EQN_lemma_worst_case_intermediate} holds. In the latter one, we have 
\begin{equation}
\y{k}=x_k-A^{k-t_k}x_{t_k}=x_k-A^{k-t_{k-1}}x_{t_{k-1}},
\end{equation}
since $t_k=t_{k-1}$. Adding and subtracting $Ax_{k-1}$ at the right hand side of the above equation, we obtain
\begin{align}
\y{k}&=x_{k}-Ax_{k-1}+A\paren{x_{k-1}-A^{k-1-t_{k-1}}x_{t_{k-1}}}\notag\\
&=w_{k}+A\y{k-1},\label{EQN_Lemma1_help}
\end{align}
where the second equality follows from the definition of $\y{k-1}$ in \eqref{EQN_Coding}. But $k-1>k_{0}$ (since $\gamma_{u,k-1}=0$), thus, the eavesdropper has intercepted $\y{k-1}$, which in turn implies 
$\y{k-1}=\expe{\y{k-1}|\info{k-1}}$ in $\event\cap \EV$ (follows from \eqref{EQN_fundamental_conditional_property} and properties of conditional expectation). Hence, from \eqref{EQN_Lemma1_help},  $\expe{\y{k}|\info{k-1}}=A\y{k-1}$ in $\event\cap \EV$, which along with \eqref{EQN_Lemma1_help} prove equation \eqref{EQN_lemma_worst_case_intermediate}.
\hfill $\qedsymbol$
 \subsection{Monotonicity results.}
  In this subsection, we prove that the case described in Lemma \ref{THM_worst_case}, where the eavesdropper receives everything after the critical event, is indeed the worst case of Theorem \ref{THM_perfect_secrecy}, in terms of confidentiality.
  To formally describe the above statement, we need to define a new channel outcome sequence $\tilde{\gamma}_{u,k}$, $\tilde{\gamma}_{k}$ that is coupled with the original outcome sequence $\gamma_{u,k}$, $\gamma_{k}$ as follows:
  \begin{equation}\label{EQN_outcome_coupling}
  \begin{aligned}
  \paren{\tilde{\gamma}_{u,k},\tilde{\gamma}_{k}}&=\paren{\gamma_{u,k},\gamma_{k}},\,\text{for all }0\le k\le k_0\\
  \paren{\tilde{\gamma}_{u,k},\tilde{\gamma}_{k}}&=\paren{\gamma_{u,k},1},\,\text{for all }k> k_0, 
  \end{aligned}
  \end{equation}
  i.e. in this new outcome sequence, the eavesdropper receives everything after time $k_0$ but the user's outcomes remain the same.
  Similarly to \eqref{EQN_batch_channel_outcomes}, we denote the batch vector of the new outcomes up to time $k$ by $\gg{k}$. 
  Next, we define again the channel outputs \eqref{EQN_channel_model}, the eavesdropper's information \eqref{EQN_eavesdropper_information} and the covariance matrix \eqref{EQN_estimate_variance_eavesdropper}, but with the original channel outcomes $\gamma_{k}$ replaced by $\tilde{\gamma}_{k}$:
  \begin{equation}\label{EQN_coupled_channel_model_information} 
  \tilde{h}_{k}=\left\{
  \begin{aligned}\y{k}, \quad &\text{ if } \tilde{\gamma}_{k}=1\\
  \varepsilon,  \quad &\text{ if } \tilde{\gamma}_{k}=0
  \end{aligned}\right.,\quad\tilde{\mathcal{I}}_k=\set{\tilde{\vct{h}}_{0:k},\vct{\gamma}_{u,0:k}},
  \end{equation}
  \begin{equation}\label{EQN_coupled_channel_covariance} 
  \tilde{P}_k=\cov{x_k|\tilde{\mathcal{I}}_k}.
  \end{equation} 
 The next lemma formally states that when the eavesdropper receives all measurements from some time $k_0$ on, then it has the smallest possible covariance matrix.
 \begin{lemma}[Comparison lemma]\label{THM_comparison_lemma}
 Let $P_k$ be the nominal error covariance matrix, as defined in \eqref{EQN_estimate_variance_eavesdropper} and $\tilde{P}_k$ be the error covariance matrix of the coupled process as defined in  \eqref{EQN_outcome_coupling}--\eqref{EQN_coupled_channel_covariance}. Then, with probability one:
 \begin{equation}\label{EQN_comparison_lemma}
 P_{k}\succeq \tilde{P}_{k},\text{ for all }k\ge 0.
 \end{equation} 
 	\hfill $\diamond$
 \end{lemma}
  
  Before we prove Lemma \ref{THM_comparison_lemma}, we prove an intermediate result.
 The following lemma compares the estimation error covariance of a Gaussian random vector $x$, given two different Gaussian information sets $\mathcal{J}_1$ and $\mathcal{J}_2$, when the first information set is smaller than the second or $\mathcal{J}_1\subseteq \mathcal{J}_2$. Intuitively, the result states that given more information we have less error. 
\begin{lemma}\label{THM_Monotonicity}
	Let $x\in\REAL^{m}$ be normal and let $\mathcal{J}_1,\,\mathcal{J}_2$ be two sets, the elements of which are vectors that depend linearly on $x$. Then,
	if $\mathcal{J}_1\subseteq \mathcal{J}_2$, we have:
	\begin{equation} \label{EQN_Data_processing_b}
	\var{x|\mathcal{J}_1}\succeq \var{x|\mathcal{J}_2}.
	\end{equation}
\end{lemma}
\begin{proof}
	Since $x$ is normal and the vectors in the sets $\mathcal{J}_1$, $\mathcal{J}_2$ depend linearly on $x$, the joint distributions of $\paren{x,\,\mathcal{J}_i}$, $i=1,\,2$ are normal. Then, it is a well know property of joint normal distributions that the conditional covariance matrices are constant or $\expe{\var{x|\mathcal{J}_1}}=\var{x|\mathcal{J}_1}$ and $\expe{\var{x|\mathcal{J}_2}}=\var{x|\mathcal{J}_2}$  \cite{anderson2005optimal}. 
	
	Next, 
recall that:
\[\cov{x|\mathcal{J}_1}=\expe{(x-\expe{x|\mathcal{J}_1})(x-\expe{x|\mathcal{J}_1})'|\mathcal{J}_1}. \]
Adding and subtracting $\expe{x|\mathcal{J}_2}$, we obtain the sum of four terms:
\[\cov{x|\mathcal{J}_1}=T_1+T_2+T_2'+T_3, \]
where:
\begin{align*}
T_1=&\expe{\Big(x-\expe{x|\mathcal{J}_2}\Big)\Big(x-\expe{x|\mathcal{J}_2}\Big)'|\mathcal{J}_1}\\
T_2=&\expe{\Big(x-\expe{x|\mathcal{J}_2}\Big)\Big(\expe{x|\mathcal{J}_2}-\expe{x|\mathcal{J}_1}\Big)'|\mathcal{J}_1}\\
T_3=&\mathrm{Cov}\Big\{\expe{x|\mathcal{J}_2}|\mathcal{J}_1\Big\}
\end{align*}
Since $\mathcal{J}_1\subseteq \mathcal{J}_2$, by the tower property:
\begin{align*}
 T_1&=\mathbb{E}\Big\{\mathbb{E}\Big[(x-\expe{x|\mathcal{J}_2})(x-\expe{x|\mathcal{J}_2})'|\mathcal{J}_2\Big]|\mathcal{J}_1\Big\}\\
 &=\mathbb{E}\Big\{\cov{x|\mathcal{J}_2}|\mathcal{J}_1\Big\}=\cov{x|\mathcal{J}_2},
\end{align*}
since $\cov{x|\mathcal{J}_2}$ is constant. Using similar arguments, we can show that $T_2=0$ since:
\begin{align*}
&\mathbb{E}\bigg\{\mathbb{E}\Big[(x-\expe{x|\mathcal{J}_2})(\expe{x|\mathcal{J}_2}-\expe{x|\mathcal{J}_1})'|\mathcal{J}_2\Big]|\mathcal{J}_1\bigg\}\\
&=\mathbb{E}\bigg\{\mathbb{E}\Big[(x-\expe{x|\mathcal{J}_2})|\mathcal{J}_2\Big](\expe{x|\mathcal{J}_2}-\expe{x|\mathcal{J}_1})'|\mathcal{J}_1\bigg\}\\
&=\expe{0|\mathcal{J}_1}=0,
\end{align*}
where the first equality holds since $\expe{x|\mathcal{J}_2}-\expe{x|\mathcal{J}_1}$ is measurable with respect to $\sigma\paren{\mathcal{J}_2}$. Finally, $T_3\succeq 0$ almost surely, which completes the proof.
\end{proof}
\begin{proof}[Proof of Lemma \ref{THM_comparison_lemma}]
	Fix some time $k\ge 0$. If $k\le k_0$ then equation \eqref{EQN_comparison_lemma} is trivially satisfied with equality. Since the channel outcomes are identical up to time $k_0$, then also $P_k=\tilde{P}_k$, for $k\le k_0$. 
	
	Suppose that $k>k_0$.
	 Fix the original channel outcomes to be $\g{k}=s=\paren{s_{u,0},\dots,s_{u,k},s_{0},\dots,s_{k}}$, for some $s\in\set{0,1}^{2k+2}$.  Then due to the coupling we have: $$\gg{k}=\tilde{s}=\paren{s_{u,0},\dots,s_{u,k},s_{0},\dots,s_{k_0},1\dots,1}.$$
	But from the definition \eqref{EQN_induced_info}, it follows that the set $\ifo{s}$ is included in $\ifo{\tilde s}$. Thus, we have $\ifo{s}\subseteq \ifo{\tilde s}$, where both sets have elements that depend linearly on the Gaussian $\x{k}$. Hence, if we apply Lemma \ref{THM_Monotonicity}, with $x=\x{k}$, $\mathcal{J}_1=\ifo{ s}$ and $\mathcal{J}_2=\ifo{\tilde s}$, we obtain that $\var{\x{k}|\ifo{ s}}\succeq \var{\x{k}|\ifo{\tilde s}}$, which also implies that the same relation holds for the $n\times n$ submatrices $\var{x_k|\ifo{ s}}\succeq \var{x_k|\ifo{\tilde s}}$. Thus, by \eqref{EQN_fundamental_conditional_property}, we have $P_k\succeq  \tilde{P}_k$ in $\set{\g{k}=s}$, for all $s\in\set{0,1}^{2k+2}$. Since the sets $\set{\g{k}=s}$ cover all the probability space, we obtain that $P_{k}\succeq \tilde{P}_k$ with probability one. 
\end{proof}
\subsection{Proof of Theorem \ref{THM_perfect_secrecy}} 
Let us first prove equation \eqref{EQN_big_o_result}. From the worst case Lemma \ref{THM_worst_case}, we obtain that $P_{k}=A^{k-k_0}P_{k_{0}}\paren{A'}^{k-k_0}$ in $\event\cap \EV$. Here, we prove that for all $k\ge k_0$:
\begin{equation}\label{EQN_main_proof_goal}
P_{k}\succeq A^{k-k_0}P_{k_{0}}\paren{A'}^{k-k_0}\text{ in }\event.
\end{equation} 
From Lemma \ref{THM_comparison_lemma}, we obtain that with probability one:
\[P_k\succeq \tilde{P}_k, \]
where $\tilde{P}_k$ is the error covariance matrix of the coupled channel outcome sequence as defined in \eqref{EQN_outcome_coupling}--\eqref{EQN_coupled_channel_covariance}.
What remains to show is that in $\event$, we have  $\tilde{P}_k=A^{k-k_0}P_{k_{0}}\paren{A'}^{k-k_0}$. Notice that since $\tilde{\gamma}_{u,k_0}=\gamma_{u,k_0}$ and $\tilde{\gamma}_{k_0}=\gamma_{k_0}$, the following events are equal: $$\tilde{\event}=\set{\tilde{\gamma}_{u,k_0}=1,\tilde{\gamma}_{k_0}=0}=\event.$$
 Also observe that since $\tilde{\gamma}_{k}=1$, $k>k_0$, the event $$\tilde{\EV}=\set{\tilde{\gamma}_{k}=1,\,\text{for all }k\ge k_0+1}=\Omega$$
is the whole probability space. Thus, applying Lemma \ref{THM_worst_case} for $\g{k}$ replaced with $\gg{k}$ and the events $\tilde{\event}$ and $\tilde{\EV}$, we obtain $\tilde{P}_{k}=A^{k-k_0}\tilde{P}_{k_{0}}\paren{A'}^{k-k_0}$, for $k\ge k_0$ in $\event$. But since $\g{k_0}$ and $\gg{k_0}$ are identical, we have $\tilde{P}_{k_0}=P_{k_0}$. 
Combining the two previous results, we 
 prove \eqref{EQN_main_proof_goal}.
 
 Next, we show that $P_{k_0}\succeq Q$ or $P_{k_0}= \Sigma_0$ in $\event.$ From the conditional expectation formula \eqref{EQN_Conditional_Recursive_formula} for $P_{k_0}$, we get that when the packet is lost  ($\gamma_{k_0}=0$), then the estimation error covariance is equal to the prediction error covariance or $P_{k_{0}}=\Sigma_{xx}$. Hence by Proposition \ref{THM_recursive}, if $k_0>0$, we have $P_{k_{0}}=AP_{k_{0}-1}A'+Q\succeq Q\succ 0$  in $\event$, while if $k_0=0$, we have $P_0=\Sigma_0\succ 0$ in $\event$.

Now, suppose that $v$ is a left-eigenvector of $A$ corresponding to the spectral radius $\rho\paren{A}$. Then, pre-multiplying and post-multiplying  by $v$ in \eqref{EQN_main_proof_goal}, we obtain: \[v'P_kv\ge \rho\paren{A}^{2\paren{k-k_0}}v'P_{k_0}v.\] 
But either $P_{k_{0}}\succeq Q$ or $P_{k_{0}}=\Sigma_0$, which implies
\[
v'P_kv\ge c \rho\paren{A}^{2\paren{k-k_0}} v'v,
\]
where $c=\min\set{\lambda_{min}(Q),\lambda_{min}(\Sigma_0)}>0$ and $\lambda_{min}(\cdot)$ denotes the smallest eigenvalue.
This, in turn, proves \eqref{EQN_big_o_result} since 
\[\Tr{P_k}\ge \lambda_{max}(P_k)\ge \frac{v'P_kv}{v'v}, \]
where $\lambda_{max}(\cdot)$ denotes the largest eigenvalue.

To prove that the coding scheme achieves perfect secrecy, notice that the user always knows $x_{t_{k}}$ and, thus, she can completely reconstruct the states $x_k$, when $\gamma_{u,k}=1$. But  this exactly implies that the condition \eqref{EQN_optimal_estimation_user} of perfect secrecy is satisfied, since $P_{u,k}=0=\cov{x_k|\vct{x}_{0:k}}$, when $\gamma_{u,k}=1$ (recall that $y_k=x_k$ in the case of state measurements). 
Finally, \eqref{EQN_big_o_result} along with the hypothesis \eqref{EQN_theorem_condition} prove that $\Tr{P_{k}}\rightarrow\infty$ with probability one.  \hfill $\qedsymbol$
\subsection{Proof of Corollary \ref{THM_rate_of_increase}}
Let $e_j$ denote the unit vector of $\REAL^{n}$ in the $j$-th direction. Then \[e_j'A^iQ(A^j)'e_j\le \rho\paren{Q} \norm{A^ie_j'}^2_2\le \rho\paren{Q} \norm{A^i}^2_2,\]
where $\norm{A^i}_2$ is the Euclidean matrix norm.
Repeating the procedure for all $j$, we obtain the inequality:
\begin{align*} 
\Tr P^{op}_k&=\sum_{i=0}^{k-1}\Tr \paren{A^i Q (A')^{i}}+\Tr\paren{ A^k\Sigma_0 (A')^k}\\
&\le n c_1 \sum_{i=0}^{k}\norm{A^k}^2_2,
\end{align*}
where $c_1=\max\set{\rho(Q),\rho(\Sigma_0)}$. But by the Gelfand's formula \cite{horn2012matrix}, $\norm{A^k}_2\sim \rho(A)^{k}$ asymptotically, which implies that there exists a constant $c_2>0$ such that:
\begin{align*}
\Tr P^{op}_k&\le c_2 \sum_{i=0}^{k}\rho(A)^{2k}=c_2\frac{\rho(A)^{2(k+1)}-1}{\rho(A)^{2}-1}.
\end{align*}
On the other hand, suppose that the first critical event occurs at $k_0$. Then, from Theorem \ref{THM_perfect_secrecy}, we obtain:
\[ \Tr P_{k}\ge c_3 \rho(A)^{2\paren{k-k_0}}, \text{ for }k\ge k_0, \]
for some constant $c_3>0$ independent of $k_0$.
Combining the previous two inequalities, we obtain:
\begin{equation}
\Tr P^{op}_k \le c \rho(A)^{-2k_0} \paren{\Tr P_k+1}
\end{equation}
for some constant $c>0$ independent of $k_0$.  
\hfill $\qed$
\subsection{Proof of Lemma \ref{THM_output_lemma}}
Consider equation \eqref{EQN_system_local_Kalman}. Since the weighted innovation sequence $K\bar{w}_k$ is white noise with covariance \[\bar{Q}=K\paren{C\bar{P}C'+R}K',\] the steps of the proof of Theorem~\ref{THM_perfect_secrecy} can be repeated.
	First, applying Lemma \ref{THM_worst_case} to system  \eqref{EQN_system_local_Kalman} with coding scheme \eqref{EQN_ouput_coding} (with $\bar{x}_k$, $\est_k$, $\estm_k$, $\bar{Q}$ instead of $x_k$, $\hat{x}_k$, $P_k$, $Q$), we obtain:
	 \begin{equation}\label{EQN_output_lemma_part1}
	 \estm_{k}=A^{k-k_0}\estm_{k_{0}}\paren{A'}^{k-k_0},\,\text{for }k\ge k_0,\text{ in } \event\cap\EV, 
	 \end{equation}
	where the events $\event,$ $\EV$ are defined in \eqref{EQN_event_B}, \eqref{EQN_event_C}.

	Second, similar to \eqref{EQN_main_proof_goal} in the proof of Theorem \ref{THM_perfect_secrecy}, the previous equality holds with inequality in $\event$ or:
	\begin{equation}\label{EQN_output_lemma_proof_help_1}
	\estm_{k}\succeq A^{k-k_0}\estm_{k_{0}}\paren{A'}^{k-k_0}\text{ in } \event.
	\end{equation}
	Third, since the critical event $\event$ occurs at $k_0$ then we have:
	\begin{enumerate}
		\item $\estm_{k_0}=A\estm_{k_0-1}A'+\bar{Q}\succeq \bar{Q}$ if $k_0>0$
		\item $\estm_{k_0}=\var{\bar{x}_0}=\bar{Q}$ if $k_0=0$.
	\end{enumerate}

	The part that is different is proving that the trace grows unbounded.
	Let $v$ be the (nonzero) left-eigenvector of $A$ corresponding to the spectral radius $\rho
	\paren{A}$. Pre-multiplying and post-multiplying  by $v$ in \eqref{EQN_output_lemma_proof_help_1}, we obtain: \[v'\estm_k v\ge \rho\paren{A}^{2\paren{k-k_0}}v'\estm_{k_0}v.\] 
	But either $\estm_{k_0}\succeq \bar{Q}$ or $\estm_{k_0}=\bar{Q}$, which implies
	\[v'\estm_k v\ge c\rho^{2\paren{k-k_0}}v'v, \]
	where 
	\[c=\frac{v'\bar{Q}v}{v'v}\ge 0.\]
	Since we have \[\Tr \estm_k\ge \lambda_{\max}(\estm_k)\ge \frac{v'\estm_k v}{v'v},\]
	we obtain \eqref{EQN_rate_lemma_output_secondary} if we show that $c$ is strictly positive or $v'\bar{Q} v> 0$. We will argue by contradiction. Suppose that $v'\bar{Q}v=0$. Substituting this in equation \eqref{EQN_DARE} we then have:
	\begin{align*} v'\bar{P}v&= \rho\paren{A}^2 v'\bar{P}v+v'Qv-\rho\paren{A}^2 v'\bar{Q}v\\
	&=\rho\paren{A}^2 v'\bar{P}v+v'Qv,
	\end{align*}
	which is a contradiction since $Q\succ 0$ and $v'Qv> 0$, while $\set{1-\rho\paren{A}^2}v'\bar{P}v\le 0$, since the system is unstable. 
	Thus, $v'\bar{Q} v> 0$ and $c>0$.
\hfill $\qed$
\subsection{Proof of Theorem \ref{THM_output}}
First we prove \eqref{EQN_output_main_result_asymptotics}. By optimality of linear estimation  we have that the estimate $\est_k$ minimizes the conditional error or:
\[\est_k=\argmin\limits_{x\in \sigma\paren{\info{k}}} \expe{\norm{\bar{x}_k-x}^2_2|\info{k}},\] 
where $x\in \sigma(\info{k})$ denotes that $x$ is measurable with respect to $\sigma(\info{k})$. Thus, if we replace $\est_k$ with $\hat{x}_k$, we obtain the following bound:
\begin{equation}
\Tr\estm_k=\expe{\norm{\bar{x}_k-\est_k}^2_2|\info{k}}\le \expe{\norm{\bar{x}_k-\hat{x}_k}^2_2|\info{k}} .
\end{equation}
Then by the inequality $\norm{a-b}_2^2\le 2\norm{a-c}_2^2+2\norm{b-c}_2^2$, we obtain:
\begin{align}
\Tr\estm_k&\le 2\expe{\norm{x_k-\hat{x}_k}^2_2|\info{k}}+ 2\expe{\norm{x_k-\bar{x}_k}^2_2|\info{k}}\nonumber\\ \label{EQN_Ouput_Proof_help_1}
&=2\Tr P_k+2\Tr \paren{\bar{P}-KC\bar{P}}.
\end{align}
The first term of the equality above holds by the equality $\Tr P_k=\expe{\norm{x_k-\hat{x}_k}^2_2|\info{k}}$ from the definition of $P_k$ in~\eqref{EQN_estimate_variance_eavesdropper}. To show the equality about the second term in \eqref{EQN_Ouput_Proof_help_1}, by the tower property:
\begin{align*}
\expe{\norm{x_k-\bar{x}_k}^2_2|\info{k}}&=\expe{\expe{\norm{x_k-\bar{x}_k}^2_2|\vct{y}_{0:k},\g{k}}|\info{k}}
\end{align*}
since $\sigma\paren{\info{k}}\subseteq \sigma\paren{\vct{y}_{0:k},\g{k}}$. But by independence $\expe{\norm{x_k-\bar{x}_k}^2_2|\vct{y}_{0:k},\g{k}}=\expe{\norm{x_k-\bar{x}_k}^2_2|\vct{y}_{0:k}}$. The right-hand expression  denotes the estimation error covariance of the Kalman filter which by assumption is assumed to be at steady state and hence equals $\bar{P}-KC\bar P$, where $\bar P$ is the steady state prediction error covariance. 
As a result,
\begin{align*}
\expe{\norm{x_k-\bar{x}_k}^2_2|\info{k}}&= \expe{\Tr \paren{\bar{P}-KC\bar{P}}|\info{k}}\\&=\Tr \paren{\bar{P}-KC\bar{P}}.
\end{align*}
which verifies \eqref{EQN_Ouput_Proof_help_1}. Finally substituting the result of Lemma~\ref{THM_output_lemma} in \eqref{EQN_Ouput_Proof_help_1} we have: 
\begin{equation*}
\Tr P_k \ge \frac{1}{2}c\rho(A)^{2\paren{k-k_0}}-\Tr\paren{\bar{P}-KC\bar{P}},\,\text{ in }\event,
\end{equation*}
for $k\ge k_0$ and some $c>0$ independent of $k_0$, which proves~\eqref{EQN_output_main_result_asymptotics}.

To prove that perfect secrecy is achieved (statement (i) of the Theorem), notice first that the user's performance is optimal. At every successful reception time $k$, the user receives $\bar{x}_k-A^{k-t_k}\bar{x}_{t_k}$ and recovers $\bar{x}_k$. As a result, the user knows $\bar{x}_k$ at the successful reception times. But from \eqref{EQN_local_Kalman_estimate}, \eqref{EQN_local_Kalman_covariance} this is exactly the optimal estimation scheme as defined in \eqref{EQN_optimal_estimation_user}.
The fact that the eavesdropper's error diverges to infinity almost surely follows from \eqref{EQN_output_main_result_asymptotics} and the hypothesis \eqref{EQN_output_hypothesis}.
\hfill $\qed$
\ifCLASSOPTIONcaptionsoff
  \newpage
\fi
\bibliographystyle{IEEEtran}
\bibliography{IEEEabrv,Paper_wireless_secrecy}
\end{document}